\documentclass[12pt]{article}

\usepackage{color}
\definecolor{darkblue}{rgb}{0.1,0.1,.7}
\usepackage[colorlinks, linkcolor=darkblue, citecolor=darkblue, urlcolor=darkblue, linktocpage]{hyperref}
\usepackage[]{amsmath}
\usepackage[]{graphicx}
\usepackage[]{latexsym}
\usepackage{slashed,graphicx,color,amsmath,amssymb}
\usepackage[mathscr]{eucal}
\usepackage{mathrsfs}
\usepackage{geometry}
\usepackage[margin=10pt,font=small,labelfont=bf]{caption}
\usepackage{amscd}
\usepackage{bm}
\usepackage{xcolor}
\usepackage{upgreek}
\usepackage[square, comma, sort&compress,numbers]{natbib}
\usepackage[all,cmtip]{xy}
\numberwithin{equation}{section}
\newcommand{\tr}{\mathrm{Tr}\,}

\newcommand{\rT}{\mathrm{T}}

\newcommand{\Res}{\mathrm{Res}}

\usepackage{proof}
\usepackage{cancel}
\newcommand{\complex}{\mathbb{C}}

\newcommand{\ket}[1]{{\left|#1\right\rangle}}

\newcommand{\skalarszorzat}[2]{{\langle #1 | #2 \rangle}}

\newcommand{\lan}{{\boldsymbol \lambda}_N}
\newcommand{\xn}{{\boldsymbol x}_N}
\newcommand{\lanp}{{\boldsymbol \lambda}^+_{N/2}}
\newcommand{\mpp}{{\boldsymbol m}_{N/2}}
\newcommand{\lN}{{\boldsymbol l}_N}
\newcommand{\aN}{{\boldsymbol a}_N}
\newcommand{\SSS}{\mathcal{S}}
\usepackage{amsthm}
\newtheorem{thm}{Theorem}
\newtheorem{prop}{Proposition}
\newcommand{\valos}{\mathbb{R}}

\begin{document}
\vspace*{-.6in} \thispagestyle{empty}
\begin{flushright}
CERN-TH-2020-020
\end{flushright}
\vspace{.2in} {\Large
\begin{center}
  {\bf  On exact overlaps in integrable spin chains  }
\end{center}}
\vspace{.2in}
\begin{center}
Yunfeng Jiang$^a$,\quad Bal\'azs Pozsgay$^b$
\\
\vspace{.3in}
\small{$^a$\textit{Theoretical Physics Department, CERN, Geneva, Switzerland\\
\vspace{.2cm}
 $^b$MTA-BME Quantum Dynamics and Correlations Research Group,\\
 Department of Theoretical Physics,\\
 Budapest University of Technology and Economics,\\
 1521 Budapest, Hungary}
}

\end{center}

\vspace{.3in}

\begin{abstract}
\normalsize{We develop a new method to compute the exact
  overlaps between integrable boundary states and on-shell Bethe states for integrable spin
  chains. Our method is based on the
  coordinate Bethe Ansatz and does not rely on the ``rotation trick'' of the corresponding lattice
  model. It leads to a rigorous proof of the factorized overlap formulae in a
  number of cases, some of which were inaccessible to earlier methods. As
  concrete examples, we consider the compact XXX and XXZ Heisenberg spin chains, and the non-compact
  $SL(2,\valos)$ spin chain.
}
\end{abstract}

\section{Introduction}
\label{sec:intro}

The overlap between an integrable boundary state and an on-shell energy eigenstate is an important
quantity in integrable models. In integrable quantum field theories, when the energy eigenstate is
the ground state, the overlap is known as the exact $g$-function. The $g$-function is a measure of
boundary degrees of freedom and is thus also called the boundary entropy. Very recently, this
quantity made its appearance in the context of AdS/CFT where it is shown
\cite{yunfeng-structure-g,Jiang:2019zig} that the structure constant of two determinant operators and one
non-BPS single trace operator at finite coupling is given by an exact $g$-function on the string world sheet.

Turning to integrable lattice models such as integrable quantum spin chains and classical
statistical lattice models, these overlaps also play an important role. They are crucial ingredients
in the context of quantum quenches, partition functions of integrable lattice models
\cite{Bajnok-unpubl}, as well as the weak coupling limit of integrability in AdS/CFT
\cite{zarembo-neel-cite1,zarembo-neel-cite3,ADSMPS2}.

The first exact result for on-shell overlaps appeared in \cite{Caux-Neel-overlap1} based on the
earlier works \cite{sajat-neel,sajat-karol}. It was found that only the Bethe states whose
rapidities are parity symmetric lead to non-vanishing overlaps. This finding was explained in
\cite{sajat-integrable-quenches}, where an integrability condition was formulated for the boundary
states. It was further argued in \cite{sajat-minden-overlaps} that it is only these integrable
states where simple factorized results can
be expected. This expectation was confirmed in all known cases (see
\cite{sajat-twisted-yangian} and references therein). It is now understood that the integrable
boundary states are closely connected to integrable boundary conditions
\cite{sajat-integrable-quenches,sajat-minden-overlaps,sajat-mps}, generalizing the seminal results
of Ghoshal and Zamolodchikov on integrable boundary QFT \cite{ghoshal-zamolodchikov}.

The exact finite volume overlap formulae have the same structure in all known cases: they are given by a product
of two parts. One part is universal and is given by
the ratio of two so-called Gaudin like determinants (which are replaced by Fredholm determinants in
the continuum limit or in the AdS/CFT situation). The other part depends on the details of the boundary state
and is a product of simple scalar factors, or a sum of such products. We note that the first work
which derived this structure was \cite{sajat-marci-boundary}, although the early results of
\cite{sajat-marci-boundary} only pertained to
integrable QFT and they were not used in the later studies of the spin chain overlaps.

The works mentioned above concern \emph{compact} spin chains, where the quantum space at each site
is finite dimensional. On the other hand, \emph{non-compact} chains with infinite dimensional local Hilbert spaces
are highly relevant in QCD and AdS/CFT. To the best of our knowledge, integrable boundary states of non-compact spin chains have never been studied before. Recently, an exact overlap formula with a specific boundary
state in a non-compact chain was
conjectured \cite{yunfeng-structure-g} in the context of AdS/CFT. The factorized overlap takes the
same form as in the
compact case. In the present paper we show that this boundary state
 is indeed integrable, and provide an actual proof for the conjectured overlap formula.

We stress that up to now there have been no methods to actually \emph{prove} the exact overlap
formulae, except for the
simplest cases in the Heisenberg spin chains which are related to the so-called diagonal $K$-matrices
\cite{Caux-Neel-overlap1}. The proof of \cite{Caux-Neel-overlap1} uses an off-shell overlap formula,
which goes back to the work of Tsushiya \cite{tsushiya} (see
also \cite{sajat-neel,sajat-karol}). It is most likely that such an off-shell formula does not exist
in other cases, which are related to off-diagonal $K$-matrices in the XXZ chain, or
any $K$-matrix in higher rank cases. The follow-up
works assumed that the structure of the factorized overlap is the same in all cases, and determined the
one-particle overlap functions using a generalization of the Quantum Transfer Matrix (QTM) method
\cite{sajat-minden-overlaps,sajat-twisted-yangian}. Alternatively, the one-particle overlap
functions could be extracted from coordinate Bethe Ansatz computations
\cite{zarembo-neel-cite1,zarembo-neel-cite3,ADSMPS2}.
And while QTM approach was rather successful in
the compact spin chain, it is not evident whether it can be generalized to the non-compact cases.

In this work we start from scratch. We develop a new method for the
rigorous proof of the overlap formulae, using only the coordinate Bethe Ansatz solution of the models. We work directly in
finite volume, and investigate certain apparent singularities of the overlaps. Our approach is a
generalization of the work of Korepin \cite{korepin-norms}, where it was rigorously proven that the norm of the
on-shell Bethe states is given by the Gaudin determinant.

The rest of the paper is structured as follows. In Section~\ref{sec:spinchains} we introduce the
local spin chains that we study in this paper and review their solution by coordinate Bethe
Ansatz. In Section~\ref{sec:boundary} we discuss integrable boundary states for these spin
chains. We also prove the boundary state proposed in \cite{yunfeng-structure-g} is indeed
integrable. We give the general strategy for the proof of exact overlap formulae using coordinate
Bethe Ansatz in Section~\ref{sec:overlap}. Concrete examples for both compact and non-compact spin
chains are presented in Section~\ref{sec:cases}. We conclude and discuss some future directions in
Section~\ref{sec:general}.

\section{Integrable local spin chains and Bethe Ansatz}
\label{sec:spinchains}
We review the definitions of various local integrable quantum spin chains and their solutions by Bethe Ansatz. More specifically, we will consider the compact XXX and XXZ spin chains and the non-compact $SL(2,\mathbb{R})$ spin chain.

\subsection{Local integrable spin chains}

We consider integrable spin chains given by local Hamiltonians
\begin{align}
\label{eq:localH}
H=\sum_{j=1}^{L} h_{j,j+1}
\end{align}
with periodic boundary condition. We denote the Hilbert space of each local site $j$ by $\mathcal{H}_j$. The dimension of $\mathcal{H}_j$ can be \emph{finite} or \emph{infinite}. Each term $h_{j,j+1}$ act on the space $\mathcal{H}_j\otimes\mathcal{H}_{j+1}$.

\paragraph{Compact spin chain} The Hamiltonian for the compact XXZ spin chain is given by
\begin{equation}
  \label{XXZ-H}
  H=\sum_{j=1}^{L}
  (\sigma^x_j\sigma^x_{j+1}+\sigma^y_j\sigma^y_{j+1}+\Delta
(\sigma^z_j\sigma^z_{j+1}-1)).
\end{equation}
where $\sigma_j^{\alpha}$ ($\alpha=x,y,z$) are the Pauli matrices. Here $\Delta$ is the anisotropy parameter. The isotropic XXX spin chain corresponds to taking $\Delta=1$. For simplicity we focus on the so-called massive regime $\Delta\ge1$ for XXZ spin chain in this paper.\par

The local Hilbert space at each site is $\mathbb{C}^2$. The two basis vectors are
\begin{align}
|\!\uparrow\rangle=\left(
                     \begin{array}{c}
                       1 \\
                       0 \\
                     \end{array}
                   \right),\qquad
|\!\downarrow\rangle=\left(
                     \begin{array}{c}
                       0 \\
                       1 \\
                     \end{array}
                   \right).
\end{align}
The isotropic XXX spin chain has $SU(2)$ symmetry. The local Hilbert spaces form the the
spin-$\tfrac{1}{2}$ representation of the $\mathfrak{su}(2)$ algebra.

\paragraph{Non-compact spin chain} Now we consider the non-compact $SL(2,\mathbb{R})$ spin chain\footnote{This spin chain is nothing but the Heisenberg XXX$_s$ spin chain with local quantum space in the non-compact $s=-1/2$ representation. We choose to call it the $SL(2,\mathbb{R})$ spin chain in accordance with the QCD and AdS/CFT literature.}
\cite{Braun:1998id,Kirch-Manashov-sl2-1}. We first introduce the $SL(2,\mathbb{R})$ algebra. The generators in the spin-$s$ representation can be written in terms of bosonic oscillators $a,a^{\dagger}$ as
\begin{align}
\label{eq:sl2algebra}
S_-=a,\qquad S_0=a^{\dagger}a+s,\qquad S_+=2s a^{\dagger}+(a^{\dagger})^2a.
\end{align}

We will focus on the spin-$\tfrac{1}{2}$ representation and take $s=1/2$ from now on. The generators satisfy the $SL(2,\mathbb{R})$ algebra
\begin{align}
[S_0,S_{\pm}]=\pm S_{\pm},\qquad [S_+,S_-]=-2S_0.
\end{align}
The local Hilbert space for this spin chain is infinite dimensional. The basis vectors are given by
\begin{align}
|n\rangle\equiv\frac{(S_+)^n}{n!}|0\rangle,\qquad n=1,2,\cdots.
\end{align}
where $|0\rangle$ is the vacuum state defined by
\begin{align}
S_-|0\rangle=0.
\end{align}
The action of the generators on the basis is given by
\begin{align}
\label{eq:actSS}
S_+|m\rangle=(m+1)|m+1\rangle,\quad S_-|m\rangle=m|m-1\rangle,\quad S_0|m\rangle=(m+\tfrac{1}{2})|m\rangle.
\end{align}
Similarly, the dual states are defined by
\begin{align}
\langle n|=\langle 0|\frac{(S_-)^n}{n!},\qquad \langle0|S_+=0
\end{align}
Using the definition of the states and the $SL(2,\mathbb{R})$ algebra, it is straightforward to show that the basis states are orthonormal
\begin{align}
\langle n|m\rangle=\delta_{m,n}.
\end{align}
The Hamiltonian takes the local form as in (\ref{eq:localH}). The local Hamiltonian density $h_{j,j+1}$ acts on $\mathcal{H}_j\otimes\mathcal{H}_{j+1}$ as
\begin{align}
\label{eq:sl2hh}
h_{j,j+1}|m_{j}\rangle\otimes|m_{j+1}\rangle=&\,\big(h(m_j)+h(m_{j+1})\big)|m_j\rangle\otimes|m_{j+1}\rangle\\\nonumber
&\,-\sum_{k=1}^{m_j}\frac{1}{k}|m_j-k\rangle\otimes|m_{j+1}+k\rangle\\\nonumber
&\,-\sum_{k=1}^{m_{j+1}}\frac{1}{k}|m_j+k\rangle\otimes|m_{j+1}-k\rangle.
\end{align}
where $h(m)$ is the harmonic sum
\begin{align}
h(m)=\sum_{k=1}^m\frac{1}{k}.
\end{align}
Like their compact cousins, non-compact spin chains also have many applications in physics. For example, the $SL(2,\mathbb{C})$ spin chain shows up in the study of Regge limit of QCD
\cite{Lipatov-1,Faddeev-Korchemsky-non-compact,QCD-sl2-review}. The $SL(2,\mathbb{R})$ spin chain which we study in this paper first appeared in the study of baryon distribution amplitudes in QCD \cite{Braun:1998id}. Later in
integrability in AdS$_5$/CFT$_4$, this Hamiltonian describes the one-loop dilatation operator of the
SL(2) sector. Recently, it also made its appearance in non-equilibrium statistical mechanics
\cite{Frassek-sl2}.

\subsection{Coordinate Bethe Ansatz}

Both the compact and non-compact spin chains are integrable and can be solved by Bethe Ansatz. We
can use either the coordinate or the algebraic Bethe Ansatz to construct the eigenstates. For our
proof below, it is more convenient to use the coordinate Bethe Ansatz. Regarding the spin-$\tfrac{1}{2}$ chains
the method goes back to the works \cite{Bethe-XXX,XXZ1,XXZ2,XXZ3}, whereas for higher spin cases it
was worked out in \cite{martins-melo-2,ragoucy-higher-spin}. In the case of the
non-compact chain we can use the results of \cite{martins-melo-2} or those of \cite{ragoucy-higher-spin} after
  analytic continuation to $s=-1/2$.

\paragraph{Reference state} The eigenstates are constructed as interacting spin waves over a proper reference state. For compact spin chain, the reference state is chosen to be the ferromagnetic vacuum
\begin{align}
|\Omega\rangle=|\!\uparrow\rangle^{\otimes L}.
\end{align}
For the non-compact spin chain, the reference state is chosen to be the Fock vacuum
\begin{align}
|\Omega\rangle=|0\rangle^{\otimes L}.
\end{align}
The reference states are eigenstates of the Hamiltonians. To obtain other eigenstates, we introduce excitations on top of the vacuum state.
A generic eigenstate is characterized by a set of rapidities $\lan\equiv\{\lambda_1,\lambda_2,\cdots,\lambda_N\}$; The corresponding eigenstate will be denoted by $\ket{\lan}$.

\paragraph{Basis vectors} Let us first introduce the basis vectors as
\begin{equation}
   \ket{x_1,\dots,x_N}\sim S^{(x_1)}_\pm S^{(x_2)}_\pm\cdots S^{(x_N)}_\pm|\Omega\rangle,
\end{equation}
where the $x_j$ denote the positions of the sites and $S_\pm^{(x_j)}$ denotes the local spin operator
at site $x_j$ that creates one excitation. Each $x_j$ runs from $1$ to $L$. From our convention of reference states, for the compact and non-compact chains the creation operators are $S^{(x)}_-$ and $S^{(x)}_+$ respectively. Now comes the crucial difference between compact and
non-compact spin chains. For the compact spin chain, we can act with $S_-^{(x_j)}$ on each site
$x_j$ \emph{only once}, thus each site can only hold one excitation. In the contrary, for non-compact spin chain,
we can act with \emph{any number} of $S_+^{(x_j)}$ on site $x_j$.

In the non-compact case the precise normalization of the basis vectors is given by
\begin{equation}
\ket{x_1,\dots,x_N}=E^{(x_1)}_+ E^{(x_2)}_+\cdots E^{(x_N)}_+|\Omega\rangle
\end{equation}
with
\begin{equation}
E_+^{(x)}\ket{m}_x=\ket{m+1}_x.
\end{equation}
The $E_+$ operators are conjugate to $S_+$, and their usage leads to a convenient representation of
the coordinate Bethe Ansatz wave functions. See \cite{ragoucy-higher-spin} for the detailed
discussion of this point.

The basis states are thus given in the two cases  by
\begin{align}
\label{eq:rangeX}
\text{Compact chain}:& \qquad|x_1,x_2,\cdots,x_N\rangle & 1\le x_1<x_2\cdots<x_N\le L,&\\\nonumber
\text{Non-compact chain}:& \qquad|x_1,x_2,\cdots,x_N\rangle & 1\le x_1\le x_2\cdots\le x_N\le L.&
\end{align}
The eigenstate $|\bm{\lambda}_N\rangle$ is given by a proper linear combination of the basis states
\begin{align}
|\bm{\lambda}_N\rangle=\sum_{\{x_j\}} \chi(\xn,\lan)\ket{x_1,x_2,\cdots,x_N},
    \end{align}
    where the range for the summation over $x_j$ are given in (\ref{eq:rangeX}).

\paragraph{Bethe wave functions} Now we discuss how to construct the wave function $\chi(\xn,\lan)$.
It takes the following form:
\begin{equation}
\label{eq:wavefunc}
\chi(\xn,\lan)
=\sum_{\sigma\in S_N}
\prod_{j>k} f(\lambda_{\sigma_j}-\lambda_{\sigma_k})
\prod_{j=1}^Ne^{i p_{\sigma_j}x_j},
\end{equation}
where $p_{\sigma_j}=p(\lambda_{\sigma_j})$ is the momentum of the excitation with rapidity $\lambda_{\sigma_j}$. $f(\lambda)$ is certain known function which is related to the $S$-matrix of excitations by
\begin{align}
S(\lambda,\mu)=\frac{f(\lambda-\mu)}{f(\mu-\lambda)}.
\end{align}
The summation in (\ref{eq:wavefunc}) is over all permutations of indices $\{1,2,\cdots,N\}$, which
is denoted by $S_N$.

Different models are distinguished by the different $p(\lambda)$ and $f(\lambda)$ functions.  For
the three spin chains under consideration, the two functions are given by
\begin{itemize}
\item Compact XXZ chain ($\Delta>1$)
\begin{align}
\label{XXZfunct}
  e^{ip(\lambda)}=\frac{\sin(\lambda-i\eta/2)}{\sin(\lambda+i\eta/2)},
  \quad f(\lambda)=\frac{\sin(\lambda+i\eta)}{\sin(\lambda)},\quad
S(\lambda)=\frac{\sin(\lambda+i\eta)}{\sin(\lambda-i\eta)},
\end{align}
where $\eta$ is related to the anisotropy by $\Delta=\cosh\eta$.
\item Compact XXX chain
\begin{align}
\label{XXXfunct}
  e^{ip(\lambda)}=\frac{\lambda-i/2}{\lambda+i/2},\quad
  f(\lambda)=\frac{\lambda+i}{\lambda},\quad
S(\lambda)=\frac{\lambda+i}{\lambda-i}.
\end{align}
\item Non-compact chain
\begin{align}
\label{eq:noncompactfunct}
e^{ip(\lambda)}=\frac{\lambda-i/2}{\lambda+i/2},\quad f(\lambda)=\frac{\lambda-i}{\lambda},\quad
S(\lambda)=\frac{\lambda-i}{\lambda+i}.
\end{align}
\end{itemize}
Our sign convention for the rapidity is such that $p'(\lambda)>0$ in all cases.

\paragraph{Bethe equations} Periodicity of the eigenstate implies that the rapidities
$\{\bm{\lambda}\}_N$ have to satisfy Bethe equations
\begin{equation}
  \label{BA0}
  e^{ip(\lambda_j)L}\prod_{k\ne j} S(\lambda_j-\lambda_k)=1.
\end{equation}
The rapidities can be found by solving Bethe equations. After finding the rapidities, the eigenvalue of the Hamiltonian is given by the total energy of the system
\begin{align}
\label{BAee}
H|\{\bm{\lambda}_N\}\rangle=E_N(\{\bm{\lambda}\}_N)|\{\bm{\lambda}_N\}\rangle,\qquad E_N(\{\bm{\lambda}\}_N)=\sum_{j=1}^N e(\lambda_j).
\end{align}
For the XXX spin chains (both compact and non-compact) the function $e(\lambda)$ is given by
\begin{align}
e(\lambda)=-\frac{2}{\lambda^2+\frac{1}{4}}.
\end{align}
For the XXZ spin chain, the function is given by
\begin{align}
e(\lambda)=\frac{4\sinh^2\eta}{\cos(2\lambda)-\cosh\eta}.
\end{align}

\paragraph{Some notations} For future use let us introduce the variables
\begin{equation}
  l_j=e^{ip(\lambda_j)}.
\end{equation}
It follows from the concrete formulae \eqref{XXXfunct}-\eqref{XXZfunct}
that $f(\lambda_j-\lambda_k)$ is a rational function of
$l_j,l_k$. With some abuse of notation we will write it as $f(l_j,l_k)$.
We can thus regard the Bethe wave function as a rational function of the $l$-variables:
\begin{equation}
  \label{Bethel}
\chi(\xn,\lan) =\sum_{\sigma\in S_N}
\prod_{j>k} f(l_{\sigma_j},l_{\sigma_k})
\prod_{j=1}^N \left(l_{\sigma_j}\right)^{x_j}.
\end{equation}
This representation will play an important role in the overlap computations. The Bethe equations are
rewritten as
\begin{equation}
  \label{BEl}
a_j
  =\prod_{k\ne j} \frac{f(l_k,l_j)}{f(l_j,l_k)},
\end{equation}
where we introduced the $a$-variables as
\begin{equation}
  \label{adef}
  a_j=l_j^L=e^{ip_jL}.
\end{equation}

\section{Integrable boundary states}
\label{sec:boundary}

In this section, we discuss integrable boundary states for integrable spin chains. We first review
the proposal of \cite{sajat-integrable-quenches} for characterizing integrable boundary states for general spin
chains. Although the proposal was motivated for compact spin chains, it is straightforward to
generalize it to the non-compact cases. On the other hand, some techniques for the explicit
constructions of the boundary states rely on the rotation trick and do not allow for an immediate
generalization to the non-compact case.

After the general discussion, we focus on explicit examples
for the compact and non-compact spin chains. The discussion for the compact cases mainly
just reviews the known results. The results on integrable boundary states of non-compact spin chains are
new. Finally we give the explicit formula for the exact overlap between a Bethe state and the
integrable state, which will be proven in later sections.

\subsection{General discussion}

\label{sec:gencriteria}

We review the definition of integrable boundary states according to
\cite{sajat-integrable-quenches}, which is inspired from the definition of boundary states in
quantum field theories \cite{ghoshal-zamolodchikov}.

Integrable models possess a family of conserved charges that are in involution with each other:
\begin{equation}
  [Q_\alpha,Q_\beta]=0.
\end{equation}
In local spin chains these charges are also local, which means they can be written in the form
\begin{equation}
  Q_\alpha=\sum_{x=1}^L q_\alpha(x).
\end{equation}
where $q_\alpha(x)$ is a local operator whose range can be chosen to be $\alpha$. In other words it
only acts on sites $x,x+1,\dots,x+\alpha-1$. The Hamiltonian of the spin chain is one of the
conserved charges, and usually we choose $H\sim Q_2$.\par

Let $\Pi$ be the space parity operator which acts on the basis vector $|i_1,i_2,\cdots,i_L\rangle$ as
\begin{align}
\Pi|i_1,i_2,\cdots,i_L\rangle=|i_L,i_{L-1},\cdots,i_1\rangle.
\end{align}
The charges can be chosen in such a way that they have fixed parity under space reflection
\begin{equation}
  \Pi Q_\alpha \Pi=(-1)^\alpha Q_\alpha,\qquad \alpha\ge 2.
\end{equation}
Integrable boundary states $\ket{\Psi}$ are defined as the elements of the Hilbert space satisfying the condition
\begin{equation}
\label{eq:intcondQQ}
Q_{2k+1}\ket{\Psi}=0,\qquad k=1,2,\dots
\end{equation}
A perhaps more natural integrability condition can be given using the transfer matrix (TM), which
generates the set of conserved charges. Such a TM can usually be
constructed systematically in the algebraic Bethe Ansatz. In the following we briefly review this
construction.

In the local integrable spin chains related to the Lie-group $G$ there is a rapidity
dependent TM $t^{\Lambda}(u)$ for all representation $\Lambda$ of $G$, such that for all $\Lambda,\Lambda'$:
\begin{equation}
  [t^{\Lambda}(u),t^{\Lambda'}(u')]=0.
\end{equation}
These transfer matrices are constructed using Lax operators as
\begin{equation}
  t^{\Lambda}(u)=\text{Tr}_a T_a^\Lambda(u),\qquad T_a(u)=\prod_{k=1}^L L^\Lambda_{ak}(u).
\end{equation}
Here $L^\Lambda_{a,k}$ are the so-called Lax operators, $k$ is the index of the local Hilbert
spaces, and $a$ stands for an auxiliary space, carrying the representation $\Lambda$ of the group $G$.

Typically there are two distinguished transfer matrices, corresponding to the cases below:
\begin{itemize}
\item $\Lambda$ is the defining representation of the group $G$. The corresponding TM will be
 called ``fundamental'' and it will be denoted as $\tau(u)$.
 \item $\Lambda$  is the representation of the physical spaces. The corresponding TM will be
 called ``physical'' and it will be  denoted as $t^0(u)$.
\end{itemize}

In our cases $G=SU(2)$. In the compact XXX case the physical spaces carry the
defining representation, therefore the two TM's mentioned above coincide. However, in the higher
spin cases and in the non-compact chain they are different.

Typically the physical TM is used to generated the local conserved charges. Expanding it in a power
series we define (see for example \cite{faddeev-how-aba-works} and \cite{Faddeev-Korchemsky-non-compact,Kirch-Manashov-sl2-1} for the non-compact cases)
\begin{align}
t^0(u)=U\,\exp\left(\sum_{n=1}^{\infty}\beta_n \frac{u^n}{n!}Q_{n+1}\right),
\end{align}
where $\beta_n$ are chosen to make the charges $Q_{n+1}$ Hermitian. $U=t^{0}(0)$ is the the
translation or shift operator.

It follows from this expansion that the integrability condition for the boundary state can be written as
\begin{equation}
  \label{intcond1}
  t^0(u) \ket{\Psi}=\Pi\,t^0(u)\,\Pi \ket{\Psi}.
\end{equation}
Several important remarks are in order.

First, this condition is somewhat stronger than \eqref{eq:intcondQQ}, because it also implies
\begin{equation}
  \label{two-site}
  U^2 \ket{\Psi}=\ket{\Psi},
\end{equation}
which does not follow from \eqref{eq:intcondQQ}. Although it has not yet been proven rigorously that
\eqref{eq:intcondQQ} implies \eqref{intcond1}, in interacting models there is no known case where
the two-site invariance \eqref{two-site} is not satisfied.

We can also require an integrability condition using the defining TM:
\begin{equation}
  \label{intcond2}
  \tau(u) \ket{\Psi}=\Pi\,\tau(u)\,\Pi \ket{\Psi}.
\end{equation}
The equivalence of \eqref{intcond2} and \eqref{intcond1} is not guaranteed. Typically the different
transfer matrices are algebraically dependent, which is established through the so-called fusion
relations (also known as the Hirota equation). In the case of $G=SU(2)$ these fusion relations
guarantee that \eqref{intcond2} and \eqref{intcond1} are equivalent, but for higher rank groups
it is possible that the
integrability conditions with TM's
corresponding to different representations have a different form \cite{gombor-private}.

We now give the explicit construction of the fundamental transfer
matrix with the $SU(2)$-symmetry, both in the compact and non-compact ones. The Lax operator at each site-$j$
is given by
\begin{align}
\label{eq:lax}
  L_{aj}(u)=u+i(\vec{\sigma}_a\cdot\vec{S}_j)=u+i\left(
  \sigma_a^zS_j^z+\sigma_a^-S_j^++\sigma_a^+S_j^-\right),
\end{align}
where it is understood that $S_\pm=S_x\pm iS_y$, and for the compact spin chain
$S^{\alpha}=\tfrac{1}{2}\sigma^{\alpha}$, whereas for the non-compact spin
chain the $S^z,S^\pm$ operators are given by (\ref{eq:sl2algebra}).

This TM satisfies a crossing relation. The Pauli matrices satisfy the relation
$\sigma^y\sigma^a\sigma^y=-(\sigma^a)^\rT$ with $a=x,y,z$, where the superscript $^\rT$ denotes
transposition. This implies
\begin{equation}
  \sigma^y_a L_{aj}(u)\sigma_a^y=-L_{aj}^{\rT_a}(-u).
\end{equation}
For the TM this means
\begin{equation}
  \tau(-u)=\Pi\tau(u)\Pi.
\end{equation}
The integrability condition is therefore equivalent to
\begin{equation}
  \label{int3}
   \tau(u) \ket{\Psi}= \tau(-u) \ket{\Psi}.
\end{equation}
We stress that this is not a generic feature of integrable models, and it is only valid for the
defining representation of the $SU(2)$-related models, and only with our specific choice for the
additive and multiplicative normalization of the local Lax operators.

In the non-compact case the integrability conditions have not yet been discussed before. We take
\eqref{intcond2} (or the equivalent conditon \eqref{int3}) as the fundamental definition of
integrability for the non-compact chain. Now we
show that this ensures the pair property for the overlaps, and thus the original condition
\eqref{eq:intcondQQ} and also \eqref{intcond1} will be satisfied.

It can be derived using the Algebraic Bethe Ansatz \cite{faddeev-how-aba-works}, that the eigenvalue
of the fundamental transfer
matrix on the Bethe state given by \eqref{eq:wavefunc} is
\begin{equation}
  \tau(u)=(u+i/2)^L\prod_{j=1}^N f(\lambda_j-u)+(u-i/2)^L\prod_{j=1}^N f(u-\lambda_j).
\end{equation}
Here we used the same notation $\tau(u)$ also for the eigenvalue.
This formula holds both in the compact XXX case and the non-compact chain, with the $f$-functions
given by \eqref{XXXfunct} and \eqref{eq:noncompactfunct}, respectively. It follows directly from the
integrability condition \eqref{int3} that the overlaps can be non-zero only when the corresponding eigenvalues
satisfy $\tau(u)=\tau(-u)$. This immediately leads to the requirement that the set $\lan$
be parity symmetric, both in the compact and non-compact cases.

\subsection{The compact chains}

In the literature two main classes of integrable states have been considered. The first class is the two-site
states which are defined as
\begin{equation}
  \ket{\Psi}=\otimes_{j=1}^{L/2} \ket{\psi},\qquad \ket{\psi}\in \complex^2\otimes\complex^2.
\end{equation}
It was shown in \cite{sajat-integrable-quenches} that in the XXX and XXZ models every two-site state
is integrable. Furthermore they correspond to integrable $K$-matrices through
\begin{equation}
  \label{psiK}
  \psi_{ab}=(K(\sigma)C)_a^b,
\end{equation}
where $C$ is constant matrix describing the so-called crossing transformation and $\sigma$ is a
special value for the rapidity parameter (for details see \cite{sajat-integrable-quenches}). The
$K$-matrix describes an
integrable boundary condition, and it satisfies the
standard Boundary Yang-Baxter (BYB) relation
\begin{equation}
  \label{BYB}
K_2(v) R_{21}(u+v) K_1(u)R_{12}(v-u)=
  R_{21}(v-u)K_1(u) R_{12}(u+v)K_2(v),
\end{equation}
where $R(u)$ is the so-called $R$-matrix in the fundamental representation, see  \cite{sajat-integrable-quenches}.

The physical meaning of the correspondence \eqref{psiK} is that an integrable boundary in space
(described by the $K$-matrix) is transformed into an integrable boundary in time (described by the boundary
state). This is the generalization of the same picture in integrable QFT, first developed by Ghoshal
and Zamolodchikov \cite{ghoshal-zamolodchikov}.

Another class of states is given by integrable matrix product states (MPS) defined as
\begin{equation}
  \label{MPSdef}
    |\Psi\rangle=\sum_{j_1,\dots,j_L=1}^2{\rm Tr}_{A}
\left[\omega_{j_L}\dots  \omega_{j_2}  \omega_{j_1}   \right]
|j_L,\dots,j_2,j_1\rangle.
\end{equation}
Here $\omega_{j}$, $j=1,2$ are matrices acting on one more auxiliary space denoted by $A$. The study
of such integrable MPS was initiated in the works \cite{zarembo-neel-cite1,zarembo-neel-cite3}, and
later it was shown in \cite{sajat-mps} that these states are also described by solutions of the BYBE, although the
corresponding $K$-matrices have an inner degree of freedom. The work  \cite{sajat-mps} also treated
two-site invariant MPS, and the two-site states above can be considered as MPS with ``trivial'', one dimensional
auxiliary space.

It was argued in \cite{sajat-minden-overlaps} that in the $SU(2)$-symmetric chains all integrable
MPS are obtained by the action of transfer matrices on two-site states. This is not true in spin
chains with higher
rank symmetries: the works \cite{sajat-mps,sajat-twisted-yangian} treated a number of ``indecomposable''
MPS's.

We note that in the higher rank cases there are two main types of integrable boundary conditions,
described by the original and the twisted BYB relations. The integrable initial states are always
related to the twisted case \cite{sajat-mps}. However, in the $SU(2)$ and $SO(N)$ related models the
two types of boundary
conditions are equivalent, which can be shown by a crossing relation, see \cite{sajat-mps} for a
detailed discussion on this issue. Here we do not treat this distinction and only refer to the
original BYB \eqref{BYB}.

\subsection{The non-compact chain}
Much less is known about integrable boundary states for non-compact spin chains compared to the
compact case. Here we present the first example which satisfy the integrability conditions. It
appears in the context of AdS/CFT \cite{yunfeng-structure-g} and an exact overlap formula has been
proposed. This integrable boundary state can be seen as a counterpart of the generalized N\'eel
state in the compact case \cite{zarembo-neel-cite2}.

\paragraph{A generalized N\'eel state} To introduce the integrable boundary state, it is more convenient to write the basis vectors of the Hilbert space as
\begin{align}
|n_1,n_2,\cdots,n_L\rangle\equiv|n_1\rangle\otimes |n_2\rangle\otimes\cdots\otimes|n_L\rangle,
\end{align}
where $n_j$ denotes the number of excitations at the site-$j$; to be more precise
\begin{equation}
  \ket{n_j}=\frac{(S_+)^{n_j}}{n_j!}\ket{0_j}.
\end{equation}
Assuming that  $L$ is even, we define a family of states which depend on a free parameter $\kappa$:
\begin{align}
\label{eq:Neelkappa}
|\text{N\'eel}_{\kappa}\rangle=\sum_{\{n_i\}}\left(\kappa^{N_{\text{odd}}}+\kappa^{N_{\text{even}}}\right)|n_1,n_2,\cdots,n_L\rangle,
\end{align}
where the summation for each $n_j$ runs over all non-negative integers. $N_{\text{odd}}$ and $N_{\text{even}}$ are the total number of excitations on odd and even sites
\begin{align}
N_{\text{odd}}=n_1+n_3+\cdots+n_{L-1},\qquad N_{\text{even}}=n_2+n_4+\cdots+n_L.
\end{align}
There are two special cases for this generalized N\'eel state. This first one is
$\kappa=1$, where $|\text{N\'eel}_{\kappa=1}\rangle$ is simply the sum over all basis vectors of
the Hilbert space. We will denote this state by $|X_F\rangle$ in what follows; it is a one-site
invariant ferromagnetic state.

The second special case is $\kappa=0$.
It follows from (\ref{eq:Neelkappa}) that the non-vanishing contributions at $\kappa=0$ are given by
$N_{\text{odd}}=0$ or $N_{\text{even}}=0$. The state $|\text{N\'eel}_0\rangle$ takes the form
\begin{align}
|\text{N\'eel}_0\rangle=\sum_{|j-k|\atop \text{even}}|\circ\cdots\circ\bullet_{j}\circ\cdots\circ\bullet_{k}\circ\cdots\rangle,
\end{align}
where the black dots stand for possible positions of excitations, and the sum is taken over all
possible distributions under the restriction that the distances between the black dots have to be
even. For example, for $L=4$, we have the following state
\begin{align}
|\text{N\'eel}_0\rangle=&\,|\circ\circ\circ\circ\rangle+|\bullet\circ\circ\circ\rangle+|\circ\bullet\circ\circ\rangle
+|\circ\circ\bullet\circ\rangle+|\circ\circ\circ\bullet\rangle\\\nonumber
&\,+|\circ\bullet\circ\bullet\rangle+|\bullet\circ\bullet\circ\rangle.
\end{align}
It is easy to see that the number of black dots cannot be larger than $L/2$.
The precise normalization for this notation is given by
\begin{align}
|\circ\rangle\equiv|0\rangle,\qquad |\bullet\rangle\equiv\sum_{n=1}^{\infty}|n\rangle.
\end{align}
Noticing that
\begin{align}
e^{S_+}|0\rangle=|0\rangle+\sum_{n=1}^{\infty}\frac{(S_+)^n}{n!}|0\rangle=|0\rangle+\sum_{n=1}^{\infty}|n\rangle=|\circ\rangle+|\bullet\rangle,
\end{align}
it is easy to see that $|\text{N\'eel}_{\kappa}\rangle$ can be written as
\begin{align}
|\text{N\'eel}_{\kappa}\rangle=\left(e^{\kappa S_+}|0\rangle\otimes e^{S_+}|0\rangle\right)^{L/2}+
\left(e^{S_+}|0\rangle\otimes e^{\kappa S_+}|0\rangle\right)^{L/2}.
\end{align}
Alternatively we can write
\begin{align}
|\text{N\'eel}_{\kappa}\rangle=e^{\kappa\mathcal{S}_+}|\Psi_{1-\kappa}\rangle,
\end{align}
where $\mathcal{S}_+=S_{+}^{(1)}+S_{+}^{(2)}+\cdots S_{+}^{(L)}$ is the $SL(2,\mathbb{R})$ generator for the full spin chain. The state $|\Psi_{\alpha}\rangle$ is defined by
\begin{align}
|\Psi_{\alpha}\rangle=\left(|0\rangle\otimes|\alpha\rangle\right)^{L/2}+\left(|\alpha\rangle\otimes|0\rangle\right)^{L/2},
\end{align}
where $|\alpha\rangle$ is the coherent state $|\alpha\rangle=e^{\alpha\,S_+}|0\rangle$. An on-shell Bethe state is the highest weight state of $SL(2,\mathbb{R})$ and hence
\begin{align}
\mathcal{S}_-|\bm{\lambda}_N\rangle=0.
\end{align}
Therefore we have
\begin{align}
\label{eq:highestweightkappa}
\langle\text{N\'eel}_{\kappa}|\bm{\lambda}_N\rangle=\langle\Psi_{1-\kappa}|e^{\kappa\mathcal{S}_-}|\bm{\lambda}_N\rangle
=\langle\Psi_{1-\kappa}|\bm{\lambda}_N\rangle
\end{align}
From the definition of $|\Psi_{\alpha}\rangle$, it is easy to see that
\begin{align}
\langle\Psi_{\alpha}|\bm{\lambda}_N\rangle=\alpha^N\,\langle\text{N\'eel}_0|\bm{\lambda}_N\rangle
\end{align}
Combing this equation with (\ref{eq:highestweightkappa}), we arrive at the following relation:
\begin{align}
\label{eq:weighedNeel}
\langle\text{N\'eel}_\kappa|\bm{\lambda}_N\rangle=(1-\kappa)^N\langle\text{N\'eel}_0|\bm{\lambda}_N\rangle,
\end{align}
where $N$ is the number of rapidities of $|\bm{\lambda}_N\rangle$.

Now we prove that $|\Psi_{\alpha}\rangle$ is indeed integrable by the criteria given in
section~\ref{sec:gencriteria}, namely the condition \eqref{intcond2} holds for it.
The strategy for the proof of integrability was developed in \cite{sajat-mps}, a closely
related method already appeared in \cite{kristjansen-proofs}.
The idea is to write both sides of
\eqref{intcond2} as a MPS, and to find a similarity transformation that connects the matrices
involved. This similarity transformation can be identified with the integrable $K$-matrix
\cite{sajat-mps}.

First of all, it is clear that
\begin{align}
\Pi|\Psi_{\alpha}\rangle=|\Psi_{\alpha}\rangle.
\end{align}
To proceed, it is useful to compute the action of the Lax operator at each site. We have
\begin{align}
L_{aj}(u)|0\rangle_j=\left(
                       \begin{array}{cc}
                         u+\tfrac{i}{2} & 0 \\
                         0 & u-\tfrac{i}{2} \\
                       \end{array}
                     \right)|0\rangle_j
+\left(
   \begin{array}{cc}
     0 & 0 \\
     i & 0 \\
   \end{array}
 \right)\,S_+^{(j)}|0\rangle_j
\end{align}
and
\begin{align}
L_{aj}(u)|\alpha\rangle_j=\left(
                     \begin{array}{cc}
                       u+\tfrac{i}{2} & i\alpha \\
                       0 & u-\tfrac{i}{2} \\
                     \end{array}
                   \right)|\alpha\rangle_j
+\left(
   \begin{array}{cc}
     i\alpha & i\alpha^2 \\
     i & -i\alpha \\
   \end{array}
 \right)\,S_+^{(j)}|\alpha\rangle_j,
\end{align}
where we have used (\ref{eq:SSalpha}) which is derived in the appendix. Taking direct product of two sites, we have
\begin{align}
\label{eq:LLtwosite}
&L_{a,j}(u)L_{a,j+1}(u)|0,\alpha\rangle_{j,j+1}=\sum_{i=1}^4 A_i|i\rangle\!\rangle_{j,j+1},\\\nonumber
&L_{a,j}(u)L_{a,j+1}(u)|\alpha,0\rangle_{j,j+1}=\sum_{i=1}^4 \tilde{A}_i|\tilde{i}\rangle\!\rangle_{j,j+1}.
\end{align}
The states are given by
\begin{align}
|1\rangle\!\rangle_{j,j+1}=&|0\rangle_j\otimes|\alpha\rangle_{j+1},& |\tilde{1}\rangle\!\rangle_{j,j+1}=&|\alpha\rangle_j\otimes|0\rangle_{j+1},\\\nonumber
|2\rangle\!\rangle_{j,j+1}=&S_+^{(j)}|0\rangle_j\otimes|\alpha\rangle_{j+1},& |\tilde{2}\rangle\!\rangle_{j,j+1}=&|\alpha\rangle_j\otimes S_+^{(j+1)}|0\rangle_{j+1},\\\nonumber
|3\rangle\!\rangle_{j,j+1}=&|0\rangle_j\otimes S_+^{(j+1)}|\alpha\rangle_{j+1},& |\tilde{3}\rangle\!\rangle_{j,j+1}=&S_+^{(j)}|\alpha\rangle_j\otimes |0\rangle_{j+1},\\\nonumber
|4\rangle\!\rangle_{j,j+1}=&S_+^{(j)}|0\rangle_j\otimes S_+^{(j+1)}|\alpha\rangle_{j+1},& |\tilde{4}\rangle\!\rangle_{j,j+1}=&S_+^{(j)}|\alpha\rangle_j\otimes S_+^{(j+1)} |0\rangle_{j+1}.
\end{align}
The matrices $A_i$ and $\tilde{A}_i$ are given by
\begin{align}
A_1=&\left(
      \begin{array}{cc}
        (u+i/2)^2 & i\alpha(u+i/2) \\
        0 & (u-i/2)^2 \\
      \end{array}
    \right),& \tilde{A}_1=&\left(
                           \begin{array}{cc}
                             (u+i/2)^2 & i\alpha(u-i/2) \\
                             0 & (u-i/2)^2 \\
                           \end{array}
                         \right),\\\nonumber
A_2=&\left(
      \begin{array}{cc}
        0 & 0 \\
        i(u+i/2) & -\alpha \\
      \end{array}
    \right),& \tilde{A}_2=&\left(
                           \begin{array}{cc}
                             -\alpha & 0 \\
                             i(u-i/2) & 0 \\
                           \end{array}
                         \right)\\\nonumber
A_3=&\left(
      \begin{array}{cc}
        i\alpha(u+i/2) & i\alpha^2(u+i/2) \\
        i(u-i/2) & -i\alpha(u-i/2) \\
      \end{array}
    \right),&
\tilde{A}_3=&\left(
               \begin{array}{cc}
                 i\alpha(u+i/2) & i\alpha^2(u-i/2) \\
                 i(u+i/2) & -i\alpha(u-i/2) \\
               \end{array}
             \right)\\\nonumber
A_4=&\left(
       \begin{array}{cc}
         0 & 0 \\
         -\alpha & -\alpha^2 \\
       \end{array}
     \right),&\tilde{A}_4=&\left(
                             \begin{array}{cc}
                               -\alpha^2 & 0 \\
                               \alpha & 0 \\
                             \end{array}
                           \right).
\end{align}

A crucial observation for our proof is that $A_i$ and $\tilde{A}_i$ are related by
\begin{align}
\label{eq:conjugateCA}
\tilde K(u)\,A_i\,\tilde K(u)^{-1}=\tilde{A}_i^{\mathrm{T}},\qquad i=1,2,3,4,
\end{align}
with the matrix $\tilde K(u)$ given by
\begin{align}
\tilde K(u)=\left(
    \begin{array}{cc}
      2u & (u+i/2)\alpha \\
      (u-i/2)\alpha & 0 \\
    \end{array}
  \right).
\end{align}
It was shown in \cite{sajat-mps} that such intertwiners can be interpreted as
integrable $K$-matrices. In fact, defining
\begin{equation}
  K(u)= \tilde K(u) \sigma^y
\end{equation}
we obtain a solution to the BYB equations \eqref{BYB}. The presence of the crossing matrix
$\sigma^y$ is a feature of the $SU(2)$-related models, see the discussion above.

Using (\ref{eq:LLtwosite}) we can write down the action of the transfer matrix on $|\Psi\rangle$ as
\begin{align}
\tau(u)|\Psi_{\alpha}\rangle=&\,\tr\left(L(u)|0\rangle\otimes L(u)|\alpha\rangle\right)^{L/2}+
\tr\left(L(u)|\alpha\rangle\otimes L(u)|0\rangle\right)^{L/2}\\\nonumber
=&\,\tr\big[A_{i_1}A_{i_2}\cdots A_{i_{L/2}}\big]|i_1,i_2,\cdots,i_{L/2}\rangle+
\tr\big[\tilde{A}_{i_1}\tilde{A}_{i_2}\cdots \tilde{A}_{i_{L/2}}\big]|\tilde{i}_1,\tilde{i}_2,\cdots,\tilde{i}_{L/2}\rangle,
\end{align}
where repeated indices are summed over from 1 to 4 and the trace is taken over the auxiliary space. The states are defined by
\begin{align}
|i_1,i_2,\cdots,i_{L/2}\rangle=&\,|i_1\rangle\!\rangle\otimes|i_2\rangle\!\rangle\otimes\cdots\otimes |i_{L/2}\rangle\!\rangle,\\\nonumber
|\tilde{i}_1,\tilde{i}_2,\cdots,\tilde{i}_{L/2}\rangle=&\,|\tilde{i}_1\rangle\!\rangle\otimes|\tilde{i}_2\rangle\!\rangle\otimes\cdots\otimes |\tilde{i}_{L/2}\rangle\!\rangle.
\end{align}
Acting the reflection operator, we obtain
\begin{align}
\label{eq:Pitau}
\Pi\,\tau(u)|\Psi_{\alpha}\rangle=&\,\tr\big[A_{i_1}\cdots A_{i_{L/2}}\big]\Pi|i_1,\cdots,i_{L/2}\rangle
+\tr\big[\tilde{A}_{i_1}\cdots\tilde{A}_{i_{L/2}}\big]\Pi|\tilde{i}_1,\cdots,\tilde{i}_{L/2}\rangle\\\nonumber
=&\,\tr\big[A_{i_1}\cdots A_{i_{L/2}}\big]|\tilde{i}_{L/2},\cdots,\tilde{i}_1\rangle
+\tr\big[\tilde{A}_{i_1}\cdots \tilde{A}_{i_{L/2}}\big]|{i}_{L/2},\cdots,{i}_1\rangle.
\end{align}
Now using the relation (\ref{eq:conjugateCA}) we can show easily
\begin{align}
&\tr\big[\tilde{A}_{i_1}\cdots\tilde{A}_{i_{L/2}}\big]=\tr\big[A_{i_{L/2}}\cdots A_{i_1}\big],\\\nonumber
&\tr\big[{A}_{i_1}\cdots{A}_{i_{L/2}}\big]=\tr\big[\tilde{A}_{i_{L/2}}\cdots \tilde{A}_{i_1}\big].
\end{align}
Plugging into the second line of (\ref{eq:Pitau}), we find
\begin{align}
\Pi\,\tau(u)|\Psi_{\alpha}\rangle=\Pi\,\tau(u)\,\Pi|\Psi_{\alpha}\rangle=\tau(u)|\Psi_{\alpha}\rangle
\end{align}
which demonstrates that the state $|\Psi_{\alpha}\rangle$ is an integrable boundary state.

\subsection{Exact overlap formulae}
The integrability condition for the boundary state $|\Psi\rangle$ leads to a number of non-trivial consequences which we discuss below.
\paragraph{Paired Bethe roots} It was first argued in \cite{sajat-integrable-quenches}, that the
condition \eqref{intcond1} imposes a strict selection rule for the overlaps between $|\Psi\rangle$
and on-shell Bethe states $|{\lan}\rangle$. Namely, the overlap
\begin{equation}
  \skalarszorzat{\Psi}{\lan}
\end{equation}
is non-zero only if the set of the Bethe roots is \emph{parity symmetric}. In the case of an even
number of particles this means that they come in pairs:
\begin{equation}
\label{eq:lambdapair}
  \{\bm{\lambda}_N\}=\{\lambda_1,-\lambda_1,\cdots,\lambda_{N/2},-\lambda_{N/2}\}.
\end{equation}
which will also be denoted as
\begin{align}
\{\bm{\lambda}_N\}=\{\bm{\lambda}_{N/2}^+,-\bm{\lambda}_{N/2}^+\}.
\end{align}
Here $\{\bm{\lambda}_{N/2}^+\}$ denotes the positive Bethe roots\footnote{In principle, it does not matter which root among the pair we call `positive'. As a convention, we can choose the one with positive real part as the positive Bethe root.}.
When the number of particles is odd, we have
\begin{align}
\{\bm{\lambda}_N\}=\{\lambda_1,-\lambda_1,\cdots,\lambda_{(N-1)/2},-\lambda_{(N-1)/2},0\}.
\end{align}
In this work we only consider overlaps with Bethe states with even numbers of particles. The cases
with odd number of Bethe roots can be treated similarly. For earlier studies with an odd number of
particles see \cite{sajat-marci-boundary,Caux-Neel-overlap2,sajat-twisted-yangian}.

\paragraph{Factorized Gaudin norm} It is well-known that the norm of the on-shell Bethe state
constructed in \eqref{Bethel} can be expressed as  \cite{korepin-norms}
\begin{equation}
  \skalarszorzat{\lan}{\lan}=
\prod_{j=1}^N \frac{1}{p'(\lambda_j)}  \prod_{j<k}^N f(\lambda_j-\lambda_k)f(\lambda_k-\lambda_j)\times
\det G,
\end{equation}
where $G$ is an $N\times N$ matrix known as the Gaudin matrix whose elements are
\begin{equation}
  G_{jk}=\delta_{jk}\left[p'(\lambda_j)L+\sum_{l=1}^N \varphi(\lambda_j-\lambda_l)\right]
-\varphi(\lambda_j-\lambda_k).
\end{equation}
The function $\varphi(\lambda)$ is defined as
\begin{equation}
  \varphi(\lambda)=-i \frac{d}{d\lambda}\log S(\lambda).
\end{equation}
The norm of an on-shell Bethe state whose rapidities are paired as in (\ref{eq:lambdapair})
factorizes further. For such symmetric states the Gaudin matrix has a block structure and the determinant can be
factorized as
\begin{equation}
  \label{Gpm}
  \det G=\det G^+\det G^-,
\end{equation}
where $G^{\pm}$ are $\tfrac{N}{2}\times\tfrac{N}{2}$ matrices with matrix elements
\begin{align}
\label{eq:Gpm}
G^\pm_{jk}=\delta_{jk}\left[p'(\lambda^+_j)L+\sum_{l=1}^{N/2} \varphi^+(\lambda_j^+,\lambda_l^+)\right]
-\varphi^\pm(\lambda^+_j,\lambda^+_k)
\end{align}
with
\begin{equation}
    \varphi^\pm(\lambda,\mu)=\varphi(\lambda-\mu)\pm \varphi(\lambda+\mu).
\end{equation}
The norm is then written as
\begin{equation}
\label{pairednorm}
  \begin{split}
&  \skalarszorzat{\lan}{\lan}=
\prod_{j=1}^{N/2} \frac{f(2\lambda_j^+) f(-2\lambda_j^+)}{(p'(\lambda^+_j))^2}
\prod_{1\le j<k\le N/2}
 \left[
\bar f(\lambda^+_j,\lambda^+_k)
\right]^2
\times
\det G^+\det G^-,
 \end{split}
\end{equation}
where we defined
\begin{equation}
  \bar f(\lambda,\mu)=f(\lambda-\mu)f(\lambda+\mu)f(-\lambda-\mu)f(-\lambda+\mu).
\end{equation}

\paragraph{Exact overlap formulae} The most important property is that the non-zero overlaps between
many integrable boundary states and on-shell Bethe states take a remarkably simple form:
\begin{equation}
  \label{ovgeneral}
\frac{\left|\skalarszorzat{\Psi}{\lan}\right|^2}{\skalarszorzat{\lan}{\lan}}
=
\prod_{j=1}^{N/2} u(\lambda_j^+)\times
  \frac{\det G^+}{\det G^-}.
\end{equation}
Here $u(\lambda)$ is the so-called one particle overlap function, which depends on the initial
state, and $G^\pm$ are the same matrices
that appeared in the factorized Gaudin norm. Below we will prove this overlap formula in a number of cases.

If the integrable boundary state is a simple product state, then all known cases involve only a single product as in (\ref{ovgeneral}). However, for other states such as the integrable MPS, the pre-factor in front
of the ratio of determinants can take more complicated forms. For more details, see the discussions in
\cite{sajat-twisted-yangian}. We put forward that our present method allows for a rigorous proof only in those
  cases when the overlap involves only a single product.

The simple form for the exact overlap formula (\ref{ovgeneral}) seems to hold for both the compact and
non-compact spin chains.
In the case of the compact chain the one-particle overlap function $u(\lambda)$ can be determined by
a ``rotation trick''  \cite{sajat-integrable-quenches,sajat-minden-overlaps}. The idea is to relate
the quantum system to a 2 dimensional classical lattice model, and to build partition functions that
are afterwards evaluated using the so called Quantum Transfer Matrix in the ``rotated channel'',
after rotating the lattice by 90$^{\circ}$.
For non-compact spin chains, the local Hilbert space at each site is infinite dimensional and the
rotation trick cannot be applied in a straightforward way. Therefore a new method is called for. Below we
develop such a method for proving the
exact overlap formula of the non-compact spin chain.


\section{Exact overlap formulae -- General strategy}
\label{sec:overlap}

In this section we explain the general strategy of our method. We postpone the concrete computations
for different integrable boundary states to Section \ref{sec:cases}.

The method is most easily
demonstrated on the compact XXX and XXZ chains, with the initial state being
\begin{equation}
  \label{XFdef}
  \ket{\Psi}=\ket{X_F}\equiv \otimes_{j=1}^L
  \begin{pmatrix}
    1 \\ 1
  \end{pmatrix}.
\end{equation}
The overlap of a given Bethe state with this state is particularly simple, because each spin
configuration has the same
weight in the overlap. The result is thus simply the sum over the wave function coefficients.

Regarding the Bethe states as given by  \eqref{Bethel},
the un-normalized overlaps are
\begin{equation}
  \label{XF1}
  \skalarszorzat{X_F}{\lan}=
 \sum_{\sigma\in S_N} \prod_{j>k} f(l_{\sigma_j},l_{\sigma_k})
 \sum_{0\le x_1<\dots <x_N\le L-1}
\prod_{j=1}^N l_{\sigma_j}^{x_j}.
\end{equation}
Such an overlap is a rational function of the set $\{l_1,\dots,l_N\}$. For this set of variables we will also use the
notation $\lN$.

We want to evaluate this rational function for the $\lN$ which satisfy the
Bethe equations \eqref{BEl}. These equations depend on $L$, therefore the first natural question is:
how do the overlaps depend on the length of the spin chain $L$?

The scalar products \eqref{XF1} carry a formal
dependence on $L$, which is hidden in the summation limits. It is our goal to make this dependence more
explicit. We will see that the summations can be performed using algebraic manipulations, such that
eventually \eqref{XF1} will be expressed as rational functions of two
sets of variables $\lN$ and $\aN=\{a_1,\dots,a_N\}$, where the $a$-variables were introduced in
\eqref{adef}.
We will see that
there will be no
further $L$-dependence. It will be this rational function where we can ``substitute the Bethe
equations'' such that the on-shell values of the overlaps can be obtained.

In order to explain the method we first consider the simplest examples.

\subsection{One-particle states}

In this case the overlap
is given by the simple sum
\begin{equation}
  \skalarszorzat{X_F}{\lambda_1}=\sum_{j=0}^{L-1} l^j_1
\end{equation}
This sum can be computed readily
\begin{equation}
  \label{ov1a}
  \skalarszorzat{X_F}{\lambda_1}=
  \begin{cases}
    L & \text{ if } l_1=1\\
  \frac{a_1-1}{l_1-1} & \text{ if } l_1\ne 1
  \end{cases}.
\end{equation}
Here we already used the new auxiliary variable $a_1=l_1^L$.

The above formulae refer to the off-shell case: they are valid for arbitrary $l_1$. Let us now
investigate the on-shell case. In the one-particle case the Bethe equation is simply
\begin{equation}
  a_1=(l_1)^L=e^{ip_1L}=1.
\end{equation}
Assuming that $l_1\ne 1$ we can substitute this into \eqref{ov1a}, and we see that the overlap
vanishes for all on-shell states with
$l_1\ne 1$. However, we will be interested in the on-shell states with non-vanishing overlap,
therefore we need to consider the case $l_1=1$.

In this simple one-particle problem the summation for the exceptional case $l_1=1$ is rather
trivial, and already given in \eqref{ov1a}. However, in order to get experience for the more
complicate cases we also derive this using a limiting procedure:
we use the continuity of the
scalar product, and investigate the $l_1\to 1$ limit of the $l_1\ne 1$ case of  \eqref{ov1a}. This gives
\begin{equation}
  \skalarszorzat{X_F}{\lambda_1=0}=\lim_{l_1\to 0}\frac{a_1-1}{l_1-1}=
\lim_{p_1\to 0}\frac{e^{ip_1L}_1-1}{e^{ip_1}-1}=
  L,
\end{equation}
where we used the definition of the $a$- and $l$-variables.

Even though this is a trivial example, it already highlights a crucial observation: having computed
a generic off-shell overlap, \emph{the operations
of ``substituting the Bethe equations'' and ``taking the limit towards the parity invariant states''
do not commute, and it is important to perform the second step first.}

\subsection{Two-particle states}

We now consider the two-particle case. The structure of the overlaps of the integrable boundary state and two-particle
  states has been studied in \cite{zarembo-neel-cite1}, where the role of the apparent pole (to be
  discussed below) was explained.

In this case the overlap is given by the summation
\begin{equation}
  \label{XF2}
  \skalarszorzat{X_F}{\lambda_1,\lambda_2}=
  f(l_{2},l_{1})
 \sum_{0\le x_1<x_2\le L-1}
 l_1^{x_1}l_2^{x_2}+
   f(l_{1},l_{2})
 \sum_{0\le x_1<x_2\le L-1}
l_2^{x_1}l_1^{x_2}.
\end{equation}
Let us now introduce the function
\begin{equation}
  B_2(l_1,l_2|L)= \sum_{0\le x_1<x_2\le L-1} l_1^{x_1}l_2^{x_2}.
\end{equation}
Assuming that
\begin{equation}
  \label{cond2}
  l_1\ne 1,\quad l_2\ne 1,\quad l_1l_2\ne 1
\end{equation}
we can perform the summation explicitly, yielding
\begin{equation}
   B_2(l_1,l_2|L)=\frac{(l_1l_2)^{L}-1}{(l_1l_2-1)(l_1-1)}-\frac{l_2^L-1}{(l_2-1)(l_1-1)}.
\end{equation}
Substituting this back into \eqref{XF2} and the introducing the $a$-variables the overlap can be written
as
\begin{equation}
  \label{XF3}
  \begin{split}
     \skalarszorzat{X_F}{\lambda_1,\lambda_2}=&
     f(l_{2},l_{1})\left[
       \frac{a_1a_2-1}{(l_1l_2-1)(l_1-1)}-\frac{a_2-1}{(l_2-1)(l_1-1)}\right]+\\
&   f(l_{1},l_{2})\left[
       \frac{a_1a_2-1}{(l_1l_2-1)(l_2-1)}-\frac{a_1-1}{(l_2-1)(l_1-1)}\right].
   \end{split}
\end{equation}

Let us now substitute the Bethe equations which in this case read
\begin{equation}
  a_1=\frac{f(l_2,l_1)}{f(l_1,l_2)},\qquad   a_2=\frac{f(l_1,l_2)}{f(l_2,l_1)}.
\end{equation}
It can be seen by direct computation that after substitution we get identically zero! This means
that all on-shell overlaps vanish, unless one of the conditions in \eqref{cond2} is broken. Note
that we did not use the specific form of the function $f(l_1,l_2)$: the vanishing of the overlap
follows directly from the functional form of the Bethe wave function.

The non-vanishing overlaps are obtained in the special cases, where $l_1=1$, $l_2=1$ or
$l_1l_2=1$. For on-shell states we can not have $l_1=1$ or $l_2=1$ except for very special cases
of fine tuned solutions. On the other hand, the condition
\begin{equation}
  l_1l_2=e^{i(p_1+p_2)}=1
\end{equation}
is very natural: this is the requirement for the pair structure in the rapidities!

In order to get the overlaps with $l_1l_2=1$ we can choose two ways: either we compute the function
$B_2$ directly for this special case, or we perform the limiting procedure from off-shell rapidities
to on-shell solutions with
$l_1l_2=1$. We choose the second method because it can be generalized to the multi-particle
cases.

If we regard the expression \eqref{XF3} as a function of 4 variables $l_1,l_2$ and $a_1,a_2$, then
it has a pole $1/(l_1l_2-1)$ associated with the pair condition. The overlap itself is a regular
function of the original $l$-variables, therefore the residue has to be zero in the physical case,
when $a_j=l_j^L$. Collecting the terms for the residue around $l_1l_2=1$ gives
\begin{equation}
  \skalarszorzat{X_F}{\lambda_1,\lambda_2}\sim
  \frac{a_1a_2-1}{l_1l_2-1}\left[
    \frac{  f(l_{2},l_{1})}{l_1-1}+  \frac{  f(l_{1},l_{2})}{l_2-1}
\right].
  \end{equation}
In the physical case $a_j=l_j^L$, and the pre-factor is a finite expression of the type $0/0$;
its finite value is actually $L$. Now we argue that the finite value of the overlap comes only
from this apparent pole: all other contributions to the overlap add up to zero for on-shell states, because they are
zero for a generic configuration satisfying \eqref{cond2}. We thus obtain the exact result for
on-shell states with the pair structure:
\begin{equation}
    \skalarszorzat{X_F}{\lambda_1,-\lambda_1}=L
\left[
    \frac{  f(l_{2},l_{1})}{l_1-1}+  \frac{  f(l_{1},l_{2})}{l_2-1}
\right],\qquad\text{with } l_2=\frac{1}{l_1}.
\end{equation}

\subsection{Multi-particle states}

The general strategy for the overlaps will mirror the one-particle case. First we introduce some
definitions and auxiliary functions.

We call a set of Bethe rapidities $\lan$ zero-free, if there is no
subset of $\lan$ where the sum of the rapidities is zero. Accordingly,
the set $\lN$ is zero free, when there is no subset of the
$l$-variables such that their product is 1. States with the
pair structure are clearly {\it not zero-free}: they are the
exceptional states that lead to non-zero overlaps.

Here we investigate overlaps with more general integrable initial states. For simplicity we still
restrict ourselves to product states, but we allow for an arbitrary two-site state, thus we consider
\begin{equation}
  \ket{\Psi}=\otimes_{j=1}^{L/2} \ket{\psi}, \qquad \ket{\psi}\in \mathcal{H}_j\otimes\mathcal{H}_{j+1}.
\end{equation}
In the XXZ chain all two-site states are integrable \cite{sajat-integrable-quenches}, but in models
with higher dimensional local spaces the integrability condition puts a restriction on
$\ket{\psi}$. Note that the one-site invariant product state considered above is a special case of
such two-site states.

The overlap with the reference state is
\begin{equation}
  \skalarszorzat{\Psi}{\Omega}=(\psi_{00})^{L/2},
  \end{equation}
where $\psi_{00}$ denotes the two-site overlap between the initial state and the reference state. In
the compact cases it is given by $\psi_{00}=\skalarszorzat{\psi}{\uparrow\uparrow}$, and in the
non-compact case by $\psi_{00}=\skalarszorzat{\psi}{00}$.

For simplicity we focus on cases where $\psi_{00}\ne 0$. Furthermore we set the normalization to
$\psi_{00}=1$, such that the overlap with the reference state is always 1. Initial states with
$\psi_{00}=0$ can be treated with a limiting procedure, see for example the case of the N\'eel state
below.

We consider the overlaps
\begin{equation}
\SSS_N(\lan)=\skalarszorzat{\Psi}{\lan}
\end{equation}
with the Bethe states given in \eqref{Bethel}.
It follows from the explicit form of the wave function that every such an overlap
is a rational function of the $l$-variables. The $L$ dependence is
hidden in the summation limits. We will show below that for
zero-free sets the
summations can be performed explicitly, yielding formulae that only involve the
$l_j$ and $a_j=(l_j)^L$ for each $j$, but they do not depend on the
volume $L$ in any other way.

Let us therefore introduce the function $\SSS_N(\lan,\aN)$, which is
obtained after these formal manipulations, and after introducing the $a$-variables:
\begin{equation}
  \SSS_N(\lN,\aN)=  \skalarszorzat{\Psi}{\lan}_{summed}.
\end{equation}
Regarded as a function of a total number of $2N$ variables, this
function does not depend on $L$ anymore. It follows from the form of the wave function and the
real space summations that these functions can always be written as
\begin{equation}
  \label{SBN}
   \SSS_N(\lN,\aN)=\sum_{\sigma\in S_N}  \prod_{j>k} f(l_{\sigma_j},l_{\sigma_k})
B_N(\sigma\lN,\sigma\aN),
\end{equation}
where $B_N$ is the ``kinematical'' part of the overlap, which arises from a simple real space
summation. It depends on the initial state; explicit formulae will be given below. In the formula
above it is understood that $\sigma\lN,\sigma\aN$ are the permutations of the corresponding ordered sets, namely
\begin{align}
\sigma\bm{l}_N=\{l_{\sigma_1},l_{\sigma_2},\cdots,l_{\sigma_N}\},\quad
\sigma\bm{a}_N=\{a_{\sigma_1},a_{\sigma_2},\cdots,a_{\sigma_N}\}.
\end{align}
The quantity $B_N$ for some special cases was already defined and computed in
  \cite{zarembo-neel-cite1}. An analogous computation for a non-integrable overlap was
  performed recently in \cite{non-int-ov}.

Let us also define the function $\tilde \SSS_N(\lN)$ which is obtained from $\mathcal{S}_N$ by the formal
substitution of the Bethe equations. This means that for each $a_j$ we
substitute the r.h.s. of the corresponding equation from
\eqref{BEl}. It is clear from the above that $\tilde \SSS_N$ is a symmetric rational function of the set $\lN$.

\begin{thm}
  \label{thm1}
The rational function $\tilde \SSS_N(\lN)$ is identically zero.
\end{thm}
\begin{proof}
The function $\tilde \SSS_N$ does not depend on the volume anymore, it only depends on the
  $l$-variables.
  In the definition of $\SSS_N$ we assumed that the set of rapidities is zero-free.
  The zero-free sets can not satisfy the integrability condition, therefore their overlaps have to
  be zero.
  This implies, that the function $\tilde \SSS_N$ vanishes for all those sets $\lN$ that are zero-free solutions
  to the Bethe equations for {\it any volume}.
This means that the rational function $\tilde \SSS_N$ vanishes at an infinite number of points,
therefore it is identically zero.
\end{proof}

The non-vanishing overlaps are obtained from $\SSS_N$ by a limiting
procedure similar to the two-particle case detailed above. The key
observation is that for each pair of rapidities (or $l$-variables $l_j$,$l_k$)
there is an apparent simple pole of $S_N$, which is proportional to
\begin{equation}
  \frac{a_ja_k-1}{l_jl_k-1}.
\end{equation}
In the physical cases, when the $a$-variables are actually given by
$a_j=(l_j)^L$, such a factor simply produces $L$. However,
it is important that we can substitute the Bethe equation only {\it
  after} these pole contributions are correctly
evaluated. Furthermore, all non-zero terms in the overlap can only
come from such terms, because if we substitute the Bethe equations {\it before}
the limit, we get zero identically.

Now we compute $\SSS_N$ for paired
rapidities. We regard $\bm{l}_N$ and $\bm{a}_N$ as
independent variables in the intermediate steps of the computation. We can
still assume that there is a well-defined function $a(l)$ connecting
the $l$- and $a$-variables, but we do not require the relation
$a(l)=l^L$ anymore. We will see below that a recursive computation of
the overlaps will require to treat more general $a(l)$ functions.

We will consider the limit
\begin{equation}
  \label{pair1}
  l_{2j-1}l_{2j}\to 1,\quad a_{2j-1}a_{2j}\to 1, \qquad j=1,\dots,N/2.
\end{equation}

Let us now investigate the apparent pole at say $l_1l_2=1$.
\begin{prop}
 The formal pole of $\SSS_N$ around the point  $l_1l_2=1$
 is of the form
 \begin{equation}
   \label{Srec}
  \SSS_N(L)\sim
 \frac{a_1a_{2}-1}{l_1l_{2}-1}
  F(\lambda_1)
   \mathop{\prod_{j=3}^{N}} f(\lambda_1-\lambda_j) f(-\lambda_1-\lambda_j)
\SSS^{\text{mod}}_{N-2}(\cancel{1},\cancel{2},L),
\end{equation}
where $\SSS^{\text{mod}}_{N-2}$ is the formal overlap for $N-2$ particles
not including 1 and 2, evaluated with the following modified
 $a$-variables:
\begin{equation}
  \label{amod}
  a^{\text{mod}}_j=
  \frac{f(l_j,l_1)}{f(l_1,l_j)}
  \frac{f(l_j,1/l_1)}{f(1/l_1,l_j)}
  a_j.
\end{equation}
In \eqref{Srec} $F(\lambda)$ is a rational function which carries the dependence
on the initial state.
\end{prop}

At present we do not have a general proof of this statement.
   However, we are able to rigorously prove it  in concrete cases. This leads to the
   determination of the function $F(\lambda)$. Examples for this will be shown in the next section.

Eq. \eqref{Srec} can be considered as a recursion relation for the
overlaps. It is rather similar to the recursion relations for scalar
products of Bethe states \cite{korepin-norms} or form factors
\cite{Smirnov-Book,sajat-nested} (see also \cite{yunfeng2,yunfeng3}). In fact, the modification rule above
is a rather straightforward generalization of a similar rule for
scalar products, first derived by Korepin in \cite{korepin-norms}.
However, the origin of the poles is different: in the previous
cases in the literature the singularities are the so-called
kinematical poles of the scalar products or form factors, which appear
when two rapidities in the bra and ket vectors approach each other. On
the other hand, here the two rapidities responsible for the pole are
within the same Bethe vector, and the apparent singularity is
associated with the pair structure. The role of such
  apparent poles was first recognized in
   \cite{zarembo-neel-cite1}, and has been used in \cite{non-int-ov} to study the large $L$
   behaviour of the overlaps.

It is important that if the original $l$- and $a$-variables satisfy the Bethe equations, then the
restricted set of $l$-variables is still on-shell with respect to the modified $a$-variables.

We now investigate the limit of the paired rapidities on the basis of the above
recursion relation. Let us therefore introduce the set of ``positive''
rapidities $\lanp$, such that the paired limit is taken as
\begin{equation}
  \label{lim2}
  \lambda_{2j-1}\to \lambda_j^+,\quad  \lambda_{2j}\to
  -\lambda_j^+,\qquad j=1\dots N.
\end{equation}
Similar notations are understood for the $l$- and $a$-variables.

For future use we introduce one more set of variables which will play
an important role. For each $j=1\dots N/2$ we define
\begin{equation}
 m_j= m(\lambda_j)=\left.-i\frac{d}{d\lambda}\log( a(\lambda))\right|_{\lambda=\lambda_j}.
  \end{equation}
In the original physical case $a_j=l_j^L=e^{ip(\lambda_j)L}$ we have
$m_j=p'(\lambda_j)L$, but generally we will treat the $m$-variables as independent.

Let us define the function $D(\lanp,\mpp^+)$ as the limit of the
function $S_N$ described by \eqref{lim2}. This is a symmetric function
under a simultaneous permutation of its variables. It is a rational function of $\lanp$ and it is
at most linear in each of the $m$-variables. The latter property follows from the fact that $S_N$ has only single poles
associated to each pair.

\begin{thm}
  The function $D$  satisfies the recursion
\begin{equation}
  \label{Drec}
  \begin{split}
& \frac{\partial  D(\lanp,\mpp^+|L)}{\partial m^+_1}=\frac{F(\lambda_1^+)}{p'(\lambda_1^+)}
\prod_{l=2}^{N/2}
\bar f(\lambda_1^+,\lambda_l^+)
\times D({\boldsymbol\lambda}^+_{N/2-1},\bm{m}^{+,mod}_{N/2-1}|L),
\end{split}
\end{equation}
where we defined
the modification rule for the  $m$-parameters
\begin{equation}
\label{mmod}
m_{mod}(\lambda)=m(\lambda)+\varphi^+(\lambda,\lambda_1^+).
\end{equation}
\end{thm}
\begin{proof}
This follows immediately from \eqref{Srec}, using also Theorem \ref{thm1}.
The modification rule for the $m$-variables follows from
\begin{equation}
   m_{mod}(\lambda)=-i\frac{d}{d\lambda}\log(a_{mod}(\lambda)),
\end{equation}
and using \eqref{amod} we get \eqref{mmod}.
\end{proof}

\begin{thm}
  The solution of the recursion \eqref{Drec} is
\begin{equation}
  \begin{split}
&D(\lanp,\mpp^+|L)= \prod_{j=1}^{N/2}\frac{F(\lambda_j^+)}{p'(\lambda_j^+)}
\prod_{1\le j<k\le N/2}\bar f(\lambda_j^+,\lambda_k^+)
\times \det G^+_{N/2}.
\end{split}
\end{equation}
\end{thm}
\begin{proof}
Our proof follows the method of Korepin derived originally for the Gaudin determinant describing the
norm of the Bethe states \cite{korepin-norms}.

First we define a function $\tilde D(\lanp,\mpp^+)$ through
  \begin{equation}
  D(\lanp,\mpp^+|L)= \prod_{j=1}^{N/2}\frac{F(\lambda_j^+)}{p'(\lambda_j^+)}
  \prod_{1\le j<k\le N/2}\bar f(\lambda_j^+,\lambda_k^+)
  \tilde D(\lanp,\mpp^+|L).
  \end{equation}
  It follows from \eqref{Drec} that the linear parts in $m_j^+$ is given by
  \begin{equation}
  \label{Drec2}
  \begin{split}
 \frac{\partial \tilde D(\lanp,\mpp^+|L)}{\partial m^+_j}=
 \tilde D({\boldsymbol\lambda}^+_{N/2-1},\bm{m}^{+,mod}_{N/2-1}|L),
\end{split}
\end{equation}
where it is understood that $m_j^+$ is not included in the arguments on the r.h.s. and the
modification rule is given by \eqref{mmod}.

The function $\tilde D$ satisfies the following properties:
\begin{itemize}
\item It is symmetric in all its variables.
\item It is at most linear in each $m_j$.
\item It is zero if all $m_j=0$.
\item The linear piece in each $m_j$ is given by \eqref{Drec2}.
\end{itemize}
It is easy to see that the unique solution for this linear recursion with the
given properties is
\begin{equation}
  \tilde D=\det \tilde G_{N/2},\qquad
  \tilde G_{jk}=\delta_{jk}\left[m_j^++\sum_{l=1}^{N/2} \varphi^+(\lambda_j^+,\lambda_l^+)\right]
-\varphi^+(\lambda^+_j,\lambda^+_k).
\end{equation}
In the physical case we need to set $m_j^+=p'(\lambda_j^+)L$.
\end{proof}

The normalized squared overlap is obtained after dividing by the norm \eqref{pairednorm}. Using the
factorization \eqref{Gpm} we eventually obtain
\begin{equation}
\frac{\left|\skalarszorzat{\Psi}{\lanp}\right|^2}{\skalarszorzat{\lanp}{\lanp}}
=
\left[\prod_{j=1}^{N/2} \frac{|F(\lambda_j^+)|^2}{ f(2\lambda_j^+) f(-2\lambda_j^+)}\right]
\frac{\det G^+}{\det G^-}.
\end{equation}

The single particle overlap function is thus determined by the function $F(\lambda)$ which
determines the apparent singularity of the off-shell overlap:
\begin{equation}
  \label{ueredm}
  u(\lambda)=\frac{|F(\lambda_j^+)|^2}{ f(2\lambda_j^+) f(-2\lambda_j^+)}.
\end{equation}

With this we have finished outlining our general strategy. What remains to be proven is the
  fundamental singularity relation \eqref{Srec}, together with finding the function $F(\lambda)$ in
  specific cases. This is presented in the next section.

\section{Exact overlap formulae -- Concrete cases}

\label{sec:cases}

\subsection{The state $\ket{X_F}$ in the Heisenberg chains}

Here we consider the state $\ket{X_F}$ defined in \eqref{XFdef}. Now the overlap can be written as
\eqref{SBN} with the $B$-function given by
\begin{equation}
  \label{BN1}
  B_N(l_1,l_2,\dots,l_N|L)=\sum_{x_1=0}^{L-N}\sum_{x_2=x_1+1}^{L-N+1}\dots \sum_{x_N=x_{N-1}+1}^{L-1}\
l_1^{x_1}l_2^{x_2}\dots l_N^{x_N}.
\end{equation}
Note that the positions of the particles go from $0$ to $L-1$.

Regarding the first function we get
\begin{equation}
  \label{B1}
  B_1(l_1|L)=\frac{l_1^L-1}{l_1-1}.
\end{equation}
The second function is determined by the simple difference equation
\begin{equation}
  \label{B2a}
  B_2(L)-B_2(L-1)=l_2^{L-1}B_1(L-1),
\end{equation}
which can be derived from the definition \eqref{BN1}. In fact, we have the following general recursion relation
\begin{align}
\label{eq:BBrecursion}
B_N(L)-B_N(L-1)=l_N^{L-1}B_{N-1}(L-1).
\end{align}
The proof of this recursion relation is as follows. Consider the chain of length $L$ with $N$
particles, the corresponding quantity is $B_N(L)$. Note that from (\ref{BN1}), $B_N(L)$ can be
written as the sum of two parts, corresponding to whether the last site is occupied or not. When the
last site is empty, all particles sit in the first $L-1$ sites, the contribution is given by
$B_N(L-1)$. When the last site is occupied, since the particle is ordered and the site can be
occupied by at most one particle, it must be the particle with $l_N$ which occupies the last
site. The contribution from this particle is $l_N^{L-1}$. The rest $N-1$ particles are distributed
in the first $L-1$ sites whose contribution is given by $B_{N-1}(L-1)$. This implies
(\ref{eq:BBrecursion}). This recursion relation is very helpful for deriving a closed form formula for
$B_N(L)$, as we will show below.

Substituting \eqref{B1} into \eqref{B2a} we get
\begin{equation}
  \label{B2b}
  B_2(L)-B_2(L-1)=l_2^{L-1}\frac{l_1^{L-1}-1}{l_1-1}=\frac{(l_1l_2)^{L-1}-l_2^{L-1}}{l_1-1}.
\end{equation}
The solution to this recursion is
\begin{equation}
  B_2(L)=\frac{(l_1l_2)^{L}}{(l_1l_2-1)(l_1-1)}-\frac{l_2^L-1}{(l_2-1)(l_1-1)}+C,
\end{equation}
where $C$ is an $L$ independent integration constant. It can be fixed easily by computing $B_2(L)$
at $L=2$, which is simply
\begin{equation}
  B_2(2)=l_2.
\end{equation}
This fixes $C$ and we get
\begin{equation}
  B_2(L)=\frac{(l_1l_2)^{L}}{(l_1l_2-1)(l_1-1)}-\frac{l_2^L}{(l_2-1)(l_1-1)}+\frac{l_2}{(l_2-1)(l_1l_2-1)}.
\end{equation}
We can continue along these lines for $N=3$. The recursion relation reads
\begin{equation}
  \label{B3a}
  B_3(L)-B_3(L-1)=l_3^{L-1}B_2(L-1).
\end{equation}
Solving this with the appropriate initial condition $B_3(3)=l_2l_3^2$ we get
\begin{equation}
  \begin{split}
  B_3(L)=
 & \frac{(l_1l_2l_3)^{L}}{(l_1l_2l_3-1)(l_1l_2-1)(l_1-1)}-\frac{(l_2l_3)^{L}}{(l_2l_3-1)(l_2-1)(l_1-1)}+\\
& + \frac{l_3^{L} l_2}{(l_3-1)(l_2-1)(l_1l_2-1)}-\frac{l_2l_3^2}{(l_1l_2l_3-1)(l_2l_3-1)(l_3-1)}.
  \end{split}
 \end{equation}
 Continuing this for the general $N$-particle case we get
 \begin{equation}
   \label{BN1sol}
   \begin{split}
     B_N(L)=
     \sum_{j=0}^N \frac{(-1)^j\left(\prod_{k=j+1}^N l_k\right)^L \prod_{k=2}^j l_k^{k-1}  }
{\prod_{k=j+1}^N \left(\prod_{o=j+1}^k l_o  -1\right)\times
\prod_{k=1}^j \left(\prod_{o=k}^j l_o-1\right)},
   \end{split}
 \end{equation}
 where we also used the general initial condition
 \begin{equation}
   B_N(N)=l_2l_3^2\dots l_N^{N-1}.
 \end{equation}
Substituting the $a$-variables leads to
 \begin{equation}
 \label{eq:finalBN}
   \begin{split}
  B_N(\{a_j\},\{l_j\},L)=
     \sum_{j=0}^N \frac{(-1)^j\prod_{k=j+1}^N a_k \prod_{k=2}^j l_k^{k-1}  }
{\prod_{k=j+1}^N \left(\prod_{o=j+1}^k l_o  -1\right)\times
\prod_{k=1}^j \left(\prod_{o=k}^j l_o-1\right)}.
   \end{split}
 \end{equation}
This will be the ingredient function for the overlaps, which will have
a summation over permutations, and multiplication with factors related
to the $S$-matrix. We emphasize that the $L$ dependence is now all hidden in
$\{a_j\}$ in the final expression (\ref{eq:finalBN}) and no longer appears in the limits of the
summations. This manipulation makes it possible to impose the Bethe equations. The formula \eqref{eq:finalBN} was first computed in  \cite{zarembo-neel-cite1}.

\subsubsection{Determining the singular piece}

Now we intend to compute the residue of the pole $1/(l_1l_2-1)$ of $\SSS_N$. The overlap itself is
given by $N!$ terms, but from the actual form of the $B$-function it can be seen that the desired
pole will only be present in those permutations that put the particles 1 and 2 to neighboring
positions. This is equivalent to the statement that $B_N$ has a pole of the form $1/(l_jl_k-1)$ if
$|j-k|=1$.

Let us therefore  pick some number $m$ and  investigate the residue
\begin{equation}
  \Res_{l_ml_{m+1}\to 1} B_N(L).
\end{equation}

To this order we write the $B$-function as
 \begin{equation}
   B_N(L)=\sum_{j=0}^N B_{N,j}(L)
 \end{equation}
 with
 \begin{equation}
   B_{N,j}(L)= \frac{(-1)^j\prod_{k=j+1}^N a_k \prod_{k=2}^j l_k^{k-1}  }
{\prod_{k=j+1}^N \left(\prod_{o=j+1}^k l_o  -1\right)\times
\prod_{k=1}^j \left(\prod_{o=k}^j l_o-1\right)}.
 \end{equation}

Let us look at the poles of the type $1/(l_ml_{m+1}-1)$.
There are two singular pieces given by $B_{N,m-1}$ and $B_{N,m+1}$,
and their sum reads
\begin{equation}
  \begin{split}
    & \frac{(-1)^{m-1}\prod_{k=m}^N a_k \prod_{k=2}^{m-1} l_k^{k-1}  }
{\prod_{k=m}^N \left(\prod_{o=m}^k l_o  -1\right)\times
  \prod_{k=1}^{m-1} \left(\prod_{o=k}^{m-1} l_o-1\right)}
+\\
& \frac{(-1)^{m-1}\prod_{k=m+2}^N a_k \prod_{k=2}^{m+1} l_k^{k-1}  }
{\prod_{k=m+2}^N \left(\prod_{o=m+2}^k l_o  -1\right)\times
  \prod_{k=1}^{m+1} \left(\prod_{o=k}^{m+1} l_o-1\right)}=\\
   & \frac{(-1)^{m-1}a_ma_{m+1}\prod_{k=m+2}^N a_k \prod_{k=2}^{m-1} l_k^{k-1}  }
   {
 \left( l_{m}  -1\right) \left(l_ml_{m+1}  -1\right)
     \prod_{k=m+2}^N \left(\prod_{o=m}^k l_o  -1\right)\times
  \prod_{k=1}^{m-1} \left(\prod_{o=k}^{m-1} l_o-1\right)}
+\\
& \frac{(-1)^{m-1}  (l_ml_{m+1})^m l_{m+1} \prod_{k=m+2}^N a_k \prod_{k=2}^{m-1} l_k^{k-1}}
{\prod_{k=m+2}^N \left(\prod_{o=m+2}^k l_o  -1\right)\times
 \left( l_ml_{m+1}-1\right) \left(l_{m+1}-1\right)
  \prod_{k=1}^{m-1} \left(\prod_{o=k}^{m+1} l_o-1\right)}.
\end{split}
\end{equation}
So altogether the singular piece in $B_N$ is
\begin{equation}
    B_N(\{a_j\},\{l_j\},L)\sim \frac{a_ma_{m+1}-1}{l_ml_{m+1}-1}\times  \frac{1}{l_m-1}
 \frac{(-1)^{m-1}\prod_{k=m+2}^N a_k \prod_{k=2}^{m-1} l_k^{k-1}  }
   {  \prod_{k=m+2}^N \left(\prod_{o=m}^k l_o  -1\right)\times
     \prod_{k=1}^{m-1} \left(\prod_{o=k}^{m-1} l_o-1\right)},
     \end{equation}
     which can be written as
     \begin{equation}
       \label{BNsing1}
         B_N(\{a_j\},\{l_j\},L)\sim \frac{a_ma_{m+1}-1}{l_ml_{m+1}-1} \frac{1}{l_m-1}
B_{N-2,m-1}(\{1,2,\dots,\cancel{m},\cancel{m-1},\dots,N\},L).
\end{equation}
In order to determine the singularity of $\SSS_N$ we need to sum over all permutations that put the
particles 1 and 2 to neighboring positions and multiply with the $f$-functions corresponding to the
permutations. It is important that once we pick positions $m$ and $m+1$ there are still two
possibilities corresponding to the relative ordering of particles 1 and 2. These two terms will have
many common factors for each $m$, and the sum of those factors which are different is
\begin{equation}
F(\lambda_m,\lambda_{m+1})=\frac{f(\lambda_{m+1}-\lambda_{m})}{l_m-1}+
\frac{f(\lambda_{m}-\lambda_{m+1})}{l_{m+1}-1}.
\end{equation}
Using the symmetry we can introduce
\begin{equation}
  \label{F1}
  F(\lambda)=\frac{f(-2\lambda)}{l(\lambda)-1}+\frac{f(2\lambda)}{l(-\lambda)-1}.
\end{equation}
The remaining additional $f$-factors for these terms will be
\begin{equation}
  \prod_{j=1}^{m-1} f(\lambda_m-\lambda_j) f(\lambda_{m+1}-\lambda_{j})
  \prod_{j=m+2}^N f(\lambda_j-\lambda_m) f(\lambda_{j}-\lambda_{m+1}).
\end{equation}
This can be written in the form
\begin{equation}
    \mathop{\prod_{j=1}^{N}}_{j\ne m,m+1} f(\lambda_m-\lambda_j) f(\lambda_{m+1}-\lambda_{j})
\times  \prod_{j=m+2}^N
\frac{f(\lambda_j-\lambda_m)}{f(\lambda_m-\lambda_j)}
\frac{f(\lambda_{j}-\lambda_{m+1})}{f(\lambda_{m+1}-\lambda_{j})}.
\end{equation}
Note that ratios of $f$-functions appear such that they multiply the $a$-variables in a well defined way,
namely the
residue can be formulated by introducing the modification rule
\begin{equation}
  a^{\text{mod}}_j=
\frac{f(\lambda_j-\lambda_m)}{f(\lambda_m-\lambda_j)}
\frac{f(\lambda_j-\lambda_{m+1})}{f(\lambda_{m+1}-\lambda_j)}
  a_j.
\end{equation}
It is important that if the original set $\lan$ satisfies the original Bethe equations, then the set
$\lan\setminus \{\lambda_m,\lambda_{m+1}\}$ satisfies the Bethe equations with the modified $a$-parameters.

Summing over all remaining permutations, altogether the singularity of the overlap at $l_1l_2=1$ is
\begin{equation}
  \SSS_N(L)\sim
 \frac{a_1a_{2}-1}{l_1l_{2}-1}
  F(\lambda_1)
   \prod_{j=3}^{N} f(\lambda_1-\lambda_j) f(-\lambda_1-\lambda_{j})
\SSS^{\text{mod}}_{N-2}(\cancel{1},\cancel{2},L)
 \end{equation}
with $F(\lambda)$ given by \eqref{F1}.

In the XXX model the functions $l(\lambda)$ and $f(\lambda)$ are given by
\eqref{XXXfunct}. Substituting them into \eqref{F1} we obtain $F(\lambda)=0$. This means that all
overlaps with $N\ne 0$ are zero. This is in agreement with the fact that any ferromagnetic state is
an eigenstate of the Hamiltonian, which lies in the $SU(2)$ multiplet of the reference state. The
overlaps of these states with any Bethe states are identically zero.

In the XXZ model the functions $l(\lambda)$ and $f(\lambda)$ are given by
\eqref{XXZfunct}. This leads to
\begin{equation}
  F(\lambda)=\frac{\sin(\lambda+i\eta/2)\sin(\lambda-i\eta/2)}{\cos^2(\lambda)}.
\end{equation}

Computing the overlap pre-factor as given by \eqref{ueredm} we get
\begin{equation}
  u(\lambda)= \frac{F^2(\lambda)}{ f(2\lambda) f(-2\lambda)}=
\tan^2(\lambda)\tan(\lambda+i\eta/2)\tan(\lambda-i\eta/2).
\end{equation}
This coincides with the result obtained in \cite{sajat-minden-overlaps}, see eq. (3.18) there.

\subsection{N\'eel and generalized N\'eel states in the Heisenberg chains}

Let us now consider the boundary state
\begin{equation}
  \ket{N_\alpha}=\otimes_{j=1}^{L/2} \left(
    \begin{pmatrix}
      1 \\ \alpha
    \end{pmatrix}
    \otimes
       \begin{pmatrix}
      1 \\ 0
    \end{pmatrix}
  \right).
\end{equation}
This state satisfies the requirement $\psi_{11}=1$, and for non-zero $\alpha$ it has finite
overlaps with all parity-invariant Bethe states. In the $\alpha\to \infty$ limit it turns
into the N\'eel state after re-scaling. It is our intention here to derive the overlaps, and also to
show that in the $\alpha\to\infty$ limit only the states with $N=L/2$ can have non-zero overlaps.

Now particles can only occupy every odd site. As a result, the computation of the kinematical sum is
almost the same as in the previous case, except that now the propagation of particles is restricted
to an even number of hoppings. As an effect, the kinematical $B_N$-function is formally the same as
before, except  for the replacement $L\to L/2$ and $l_j\to l_j^2$ for each $j=1,\dots,N$. Also, the
overlap receives an overall factor of
$\alpha^N$. As an effect of these changes, instead of the direct pole of the type $1/(l_1l_2-1)$ we obtain poles
\begin{equation}
  \frac{1}{l_1^2l_2^2-1}= \frac{1}{l_1l_2-1}   \frac{1}{l_1l_2+1}.
\end{equation}
It can be seen that the residue at $l_1l_2=1$ gets an extra factor of $1/2$. Putting these
modifications together we can extract the $F$-function as
\begin{equation}
  F(\lambda)=\frac{\alpha^2}{2}\left[\frac{f(-2\lambda)}{l^2(\lambda)-1}+\frac{f(2\lambda)}{l^2(-\lambda)-1}\right].
\end{equation}
In the XXX model the substitution of  \eqref{XXXfunct} leads to
\begin{equation}
  F(\lambda)=\alpha^2\frac{u^2+1/4}{4u^2}.
\end{equation}
Altogether the one-particle overlap function with the un-normalized state becomes
\begin{equation}
  u(\lambda)= \frac{F^2(\lambda)}{ f(2\lambda) f(-2\lambda)}=
 \alpha^4 \frac{u^2+1/4}{16u^2}.
\end{equation}
In order to obtain the overlaps with the N\'eel state we need to perform the limit $\alpha\to
\infty$ after re-scaling by $\alpha^L$. It follows immediately that only the overlaps with $N=L/2$
survive, as expected.

The resulting overlap formula agrees with the earlier results
\cite{Caux-Neel-overlap1,zarembo-neel-cite1,zarembo-neel-cite2}.

\subsection{Generalized N\'eel state in the $SL(2,\mathbb{R})$ chain}

Let us first consider the overlap with the generalization of $\ket{X_F}$, namely a one-site invariant
state
\begin{equation}
  \ket{X_F}=\otimes_{j=1}^L
  \begin{pmatrix}
    1 \\ 1\\ \vdots
  \end{pmatrix}=\otimes_{j=1}^L \big(e^{S_+}|0\rangle\big)=e^{\mathcal{S}_+}|\Omega\rangle
\end{equation}
This state was already introduced in Section \ref{sec:gencriteria} as the special case of the
generalized N\'eel state $\ket{\text{N\'eel}_{1}}$.
This vector belongs to the multiplet of the reference state, so the overlaps with the Bethe states will
vanish, in accordance with relation \eqref{eq:weighedNeel} for $\kappa=1$. However, it is useful to compute the associated kinematical functions, which can be used
later for general $\kappa$.

The overlap with $\ket{X_F}$ is given by the same form as in \eqref{SBN} but now the kinematical sum is
\begin{equation}
  \label{BN2}
  B_N(l_1,l_2,\dots,l_N|L)=\sum_{x_1=0}^{L-1}\sum_{x_2=x_1}^{L-1}\dots \sum_{x_N=x_{N-1}}^{L-1}\
l_1^{x_1}l_2^{x_2}\dots l_N^{x_N}.
\end{equation}
The difference from the compact XXZ model is that now an arbitrary number of particles
can occupy the same site, and this changes the summation limits.

In the one-particle case we get the same formula as before:
\begin{equation}
  \label{B1a}
  B_1(l_1|L)=\frac{l_1^L-1}{l_1-1}.
\end{equation}
For $N=2$ the relevant recursion relation is
\begin{equation}
    B_2(L)-B_2(L-1)=l_1\frac{(l_1l_2)^{L-1}}{l_1-1}- \frac{l_2^{L-1}}{l_1-1}.
\end{equation}
The initial condition is $B_2(1)=1$. The solution satisfying this condition is
\begin{equation}
  B_2(L)=l_1\frac{(l_1l_2)^{L}}{(l_1l_2-1)(l_1-1)}- \frac{l_2^{L}}{(l_1-1)(l_2-1)}+\frac{1}{(l_1l_2-1)(l_2-1)}.
\end{equation}
Regarding the general multi-particle case the difference equation is
\begin{equation}
  B_N(L)-B_N(L-1)=l_N^{L-1}B_{N-1}(L)
\end{equation}
with the initial condition
\begin{equation}
 B_N(1)=1.
\end{equation}
The general solution is
\begin{equation}
  \label{BN2sol}
   \begin{split}
     B_N(L)=
     \sum_{j=0}^N \frac{(-1)^j \prod_{k=j+1}^N l_k^L  l_k^{N-k}  }
{\prod_{k=j+1}^N \left(\prod_{o=j+1}^k l_o  -1\right)\times
\prod_{k=1}^j \left(\prod_{o=k}^j l_o-1\right)}.
   \end{split}
 \end{equation}

The analysis of the singularity of $B_N$ can be performed in a similar way as before. We get the relation
\begin{equation}
         B_N(\{a_j\},\{l_j\},L)\sim \frac{a_ma_{m+1}-1}{l_ml_{m+1}-1} \frac{l_m}{l_m-1}
B_{N-2,m-1}(\{1,2,\dots,\cancel{m},\cancel{m-1},\dots,N\},L).
\end{equation}
The only change compared to \eqref{BNsing1} is the appearance of an extra factor of
$l_m$. Completing the computation we obtain the $F$-function as
\begin{equation}
  \label{Fsl2a}
  \begin{split}
  F(\lambda)&=\frac{f(-2\lambda)l(\lambda)}{l(\lambda)-1}+\frac{f(2\lambda)l(-\lambda)}{l(-\lambda)-1}=\\
&=\frac{f(-2\lambda)}{1-l(-\lambda)}+\frac{f(2\lambda)}{1-l(\lambda)}.
\end{split}
\end{equation}
In the $SL(2,\mathbb{R})$ case the corresponding functions are given by \eqref{eq:noncompactfunct}.
Substituting them into \eqref{Fsl2a} we get $F(\lambda)=0$ as expected.

Now we consider the generalized N\'eel state $\ket{\text{N\'eel}_0}$. The difference is once again that we
need to perform the change $l_j\to l_j^2$. This leads eventually to
\begin{equation}
F(\lambda)=\frac{1}{2}\left[  \frac{f(-2\lambda)}{1-l^2(-\lambda)}+\frac{f(2\lambda)}{1-l^2(\lambda)}\right].
\end{equation}
and the overlap function becomes
\begin{equation}
  u(\lambda)= \frac{F^2(\lambda)}{ f(2\lambda) f(-2\lambda)}.
\end{equation}
Substituting \eqref{eq:noncompactfunct} we get the same result as in the XXX case:
\begin{equation}
  u(\lambda)= \frac{F^2(\lambda)}{ f(2\lambda) f(-2\lambda)}=
  \frac{u^2+1/4}{16u^2}.
\end{equation}
This result agrees with the findings of \cite{yunfeng-structure-g}.

\section{Conclusions and discussions}
\label{sec:general}

We presented a new method to derive and prove exact overlap formulae in integrable spin chains. The
method is based on the coordinate Bethe Ansatz representation of the wave functions. The key
identity is the singularity property \eqref{Srec} of the off-shell overlaps. This is a new result of
the present work, which leads to the proof of exact overlaps in a number of cases presented
in Section \ref{sec:cases}.

It is important to compare the present method to the previous derivation of
\cite{Caux-Neel-overlap1}, which was the only available rigorous proof before our work. The paper
\cite{Caux-Neel-overlap1} derived the factorized overlaps starting from an exact off-shell
determinant formula, based on \cite{sajat-neel,sajat-karol} and going back to the work of Tsushiya \cite{tsushiya}. This method only works for the boundary states corresponding to the so-called
diagonal $K$-matrices. On the other hand, our method is applicable even for off-diagonal
$K$-matrices, when there is no determinant formula for the off-shell overlaps.

Nevertheless our method has its drawbacks and limitations. First of all, we were not able to provide
a general proof of the relation \eqref{Srec}, we only proved it on a case by case basis. Clearly, it
would be important to find the deeper reason why such a relation holds. Second, our method relies
heavily on the coordinate Bethe Ansatz, and therefore it can not be applied in situations where this
method fails, for example in models with $U(1)$-symmetry breaking. It would be desirable to study
the same problems in more general frameworks such as the Separation of Variables (SoV) method. Such
a future study might also be helpful for studying overlaps in models solvable by the nested Bethe ansatz, where the present
method seems rather cumbersome. We plan to return to this question in future work.

Regarding the interpretation of
the factorized overlap formulae let us mention once more the work
\cite{sajat-marci-boundary}, which treated excited state $g$-functions in integrable QFT. These
objects are completely analogous to the finite volume overlaps in the spin chain. In
\cite{sajat-marci-boundary} the known structure \eqref{ovgeneral} of the overlaps was derived, even
before the analogous results for spin chains appeared in the context of the quantum quench. The
work \cite{sajat-marci-boundary} compared  the computation of certain time-dependent
one-point functions in finite and infinite volumes, and derived the correct ratio of determinants
using only the density of states for the restricted, parity symmetric configurations. Therefore,
\cite{sajat-marci-boundary} provides a rather natural interpretation for the overlaps, much like the
parallel observation that the original Gaudin-determinant describes both the density of states and
the norm of the Bethe states (see also \cite{XXZ-gaudin-norms}). It would be desirable to work out
the arguments of \cite{sajat-marci-boundary} also in the spin chain situation, and to make them
precise. This would complete the understanding of the factorized overlap formulae.

Finally we note  that the our method can be applied directly to the
Lieb-Liniger model to derive the overlaps with the BEC state, originally found in
\cite{caux-stb-LL-BEC-quench} and proven by a scaling limit of the spin chain in \cite{Brockmann-BEC}.

\subsection*{Acknowledgments}

The work of B.P.
was partially supported by the National Research Development and Innovation Office (NKFIH) of Hungary under
grant  K-16 No.~119204,
by the
J\'anos Bolyai Research Scholarship of the Hungarian
Academy of Sciences and the
\'UNKP-19-4 New National Excellence Program of the Ministry for Innovation and Technology.

\appendix
\section{Formula for coherent states}
In this appendix, we collect some formula for coherent states which are useful in the main text. From the $SL(2,\mathbb{R})$ algebra and using the formula
\begin{align}
e^A\,B\,e^{-A}=B+[A,B]+\frac{1}{2!}[A,[A,B]]+\cdots
\end{align}
we can prove the following results
\begin{align}
e^{-\alpha\,S_+}\,S_-\,e^{\alpha\,S_+}=&\,S_-+2\alpha S_0+\alpha^2 S_+,\\\nonumber
e^{-\alpha\,S_+}\,S_0\,e^{\alpha\,S_+}=&\,S_0+\alpha\,S_+.
\end{align}
Using these relations, we can prove the action of generators on the coherent state
\begin{align}
\label{eq:SSalpha}
S_0|\alpha\rangle=\frac{1}{2}|\alpha\rangle+\alpha\,S_+|\alpha\rangle,\qquad S_-|\alpha\rangle=\alpha|\alpha\rangle+\alpha^2\,S_+|\alpha\rangle.
\end{align}

\providecommand{\href}[2]{#2}\begingroup\raggedright\endgroup


\end{document}